\documentclass[aps,twocolumn,pra,superscriptaddress,nofootinbib]{revtex4-2}

\usepackage{amsmath}

\usepackage{amssymb}

\usepackage{booktabs}
\usepackage{longtable}
\usepackage{array}

\newcolumntype{P}[1]{>{\raggedright\arraybackslash}p{#1}}

\usepackage{mathrsfs}
\usepackage{bm, dsfont}
\usepackage[usenames,dvipsnames]{color}
\usepackage{colortbl}
\usepackage[table]{xcolor}
\usepackage{enumitem}

\usepackage{physics}

\usepackage{graphicx}
\usepackage{natbib}
\usepackage[colorlinks=true,linkcolor=blue,citecolor=blue,urlcolor=blue]{hyperref}
\usepackage{cleveref}
\usepackage{hypcap}
\usepackage{verbatim, float}
\usepackage{psfrag}
\usepackage{tikz}
\usetikzlibrary{arrows.meta, automata,
                positioning,
                quotes}
\usepackage[normalem]{ulem}
\usepackage{physics}		% for braket, etc.
\usepackage{amsmath}
\usepackage{amsthm}

\newtheorem{lemma}{Lemma}
\newtheorem{theorem}{Theorem}

\newtheorem{definition}{Definition}
\newtheorem{thm}{Theorem}
\newtheorem{lem}{Lemma}
\newtheorem{prop}{Proposition}
\newtheorem{defn}{Definition}
\newtheorem{corollary}{Corollary}
\newtheorem{remark}{Remark}
\usepackage{mathtools}

\DeclarePairedDelimiter{\of}{\lparen}{\rparen}

\newcommand{\sumab}[2]{\underset{#1}{\overset{#2}{\sum}}}
\newcommand{\suma}[1]{\underset{#1}{\sum}}
\newcommand{\set}[1]{\left\{#1\right\}}

\newcommand{\brakett}[2]{\langle #1 \vert #2 \rangle}

\newcommand{\mA}{\mathcal{A}} 
\newcommand{\mU}{\mathcal{U}}

\newcommand{\mF}{\mathcal{F}}

\newcommand{\xp}[1]{^{\otimes #1}}

\let\S\relax
\DeclareMathOperator{\S}{S} % symmetric group
\DeclareMathOperator{\C}{\mathbb{C}} % complex numbers
\DeclareMathOperator{\End}{End} 

\DeclareMathOperator{\U}{U}

\begin{document}

\author{Riccardo Castellano}
\affiliation{Department of Applied Physics, University of Geneva, Switzerland}
\author{Dmitry Grinko}
\affiliation{QuSoft, Amsterdam, The Netherlands}
\affiliation{Institute for Logic, Language and Computation, University of Amsterdam, The Netherlands}
\affiliation{Korteweg-de Vries Institute for Mathematics, University of Amsterdam, The Netherlands}
\author{Sadra Boreiri}
\affiliation{Department of Applied Physics, University of Geneva, Switzerland}
\author{Nicolas Brunner}
\affiliation{Department of Applied Physics, University of Geneva, Switzerland} 
\author{Jef Pauwels}
\affiliation{Department of Applied Physics, University of Geneva, Switzerland} 
\affiliation{Constructor University, 28759 Bremen, Germany}

\title{Stronger Welch Bounds and Optimal Approximate $k$-Designs}

\begin{abstract}

{The question of how well a finite ensemble of quantum states can approximate Haar randomness is relevant in many areas of quantum physics. Here we consider this problem from a geometrical perspective, via the notion of Welch bounds. The latter are a set of inequalities, originally derived in the context of signal processing, which precisely characterize exact $k$-designs, i.e. a finite ensemble of states that exactly reproduces the $k$-th moment of the Haar measure. We derive stronger Welch bounds which enable a quantitative and operationally meaningful characterisation of approximate $k$-designs. Our new bounds are often tight, allowing us to characterize optimal approximate designs in regimes where the cardinality of the ensemble is small. A notable example is complete sets of mutually unbiased bases (MUBs), which form an optimal approximate 3-design among 1-designs of their cardinality. This leads to a new variational approach to determine the existence of complete sets of MUBs, which we explore numerically in dimension six. We highlight a connection to shadow tomography, relating our bounds to the shadow norm. We comment on how the method extends to other types of designs, including toric and unitary designs.} A key technical ingredient, is the computation of the complete spectrum of the partially transposed symmetric-subspace projector, including multiplicities and eigenvectors, which may find further applications.

\end{abstract}
\maketitle

\section{Introduction}

How uniformly can a finite set of unit vectors be distributed in a {complex vector space}? This geometric question arises across mathematics, physics, and engineering. {It is exemplified by the well-known Thomson problem, where one aims to place charged particles on the unit sphere as far apart from each other as possible, in order to minimize the total electrostatic potential energy. The question is also central in information theory, notably in coding theory, signal design, and compressed sensing. When signals are represented by complex vectors, there is a fundamental trade-off between the number of signals, the dimension of the vector space, and the cross-correlations between the signals. Formally, this is captured by a family of inequalities derived by Welch in 1974}~\cite{WelchBound}{, commonly referred to as Welch bounds, which have been extensively studied}~\cite{WelchBoundGeom(2012),Delsarte1977,haikin2018,Ulukus2002}.

{Welch bounds are also directly relevant to quantum information, where many problems require the construction of sets of states that are as uniformly distributed as possible. Consider a set of $N$ pure quantum states, represented by complex unit vectors in a Hilbert space of dimension $d$:} $\chi:=\{\ket{\psi_i}\}_{i=1}^N\subset\mathbb C^d$. The {Welch inequality of order $k$} reads
\[
  \sum_{i,j=1}^N \bigl|\braket{\psi_i}{\psi_j}\bigr|^{2k}
  \;\ge\; \frac{N^2}{\binom{d+k-1}{k}},\qquad k\in\mathbb N^+.
\]
{The quantity on the left-hand side}, called the {unnormalized} $k$-frame potential, {measures moments of the pairwise overlaps. For each $k$, it is bounded from below by a constant depending on $d$ and $N$.}

{Of particular interest are cases in which the Welch inequality is saturated.} For $k=1$, equality is attained by tight frames{, whose rescaled projectors define a generalized quantum measurement given by a positive-operator-valued measure (POVM) with rank-one elements. More generally, tight frames are spanning sets of vectors that may be overcomplete, with redundancy that can be exploited to enable stable data recovery over noisy or lossy signal transfer.}

{For a general $k$, the Welch inequality is saturated solely by sets of states featuring a remarkable property, namely forming a $k$-design \cite{klappenecker2005,Roy2007}. The latter} have become central in quantum information, {as} they provide finite {and highly symmetric} surrogates of Haar-random states \cite{ambainis2007}. {More precisely, a complex projective $k$-design is a finite ensemble of pure states whose average reproduces the Haar average of every polynomial of degree $k$ in the state amplitudes and degree $k$ in their complex conjugates. The Haar measure describes the uniform distribution over all pure states. State designs play a key role in quantum-state estimation problems \cite{Zhu2022}. The most well-known examples are: (i)} complete sets of {mutually unbiased bases (MUBs) with} $N=d(d+1)$ {(for prime-power dimensions) \cite{woottersMUB} and (ii) symmetric informationally complete (SIC) POVMs with} $N=d^2$ {elements} \cite{renes2004symmetric}. Both families are {$2$-designs, and are} central in optimal quantum state estimation~\cite{woottersMUB,Zhu2022}, quantum cryptography~\cite{Cerf2002,Tavakoli2020}, entanglement detection~\cite{Spengler2012}, and quantum foundations~\cite{Fuchs2017}.

{The saturation of Welch bounds thus provides a complete characterization of finite sets of states reproducing Haar moments up to order $k$, forming exact designs. This requires, however, very specific parameters; in particular, the cardinality $N$ of the set must be large enough, and the required cardinality grows very quickly with the dimension $d$ and $k$.}

{Beyond this idealized case, it is natural to ask whether the Welch bounds can still be used, notably for characterizing how uniform and how random more general sets of states are. As a first step, consider sets of states that are close to forming an exact design, so-called approximate designs \cite{ambainis2007}, which have recently found applications in shadow tomography \cite{Huang2020Predicting,Innocenti2023} and the dynamics of many-body quantum systems \cite{Cottler2023}. By continuity, an approximate design nearly saturates the corresponding Welch bound.
However, the standard bound does not determine how close to saturation
a set of prescribed cardinality can come. When the cardinality is too
small for an exact design, an unavoidable gap remains, which the
standard Welch bound does not quantify.}

This raises several natural questions:  
(i) can Welch bounds be strengthened{? (ii) in particular, can they} remain meaningful for cardinalities below those required for $k$-designs?
{(iii)} can these ideas be used to quantify how well a given set of states approximates a $k$-design? {(iv)} can one identify provably optimal approximate $k$-designs?

We answer these questions affirmatively. 
First, we derive a strengthened family of inequalities that strictly improve upon the standard Welch bounds. These {stronger Welch} bounds turn out to be tight in regimes where the standard bounds are not. Second, we introduce a natural quantitative measure of how closely a finite {set of states} approximates a $k$-design, and we obtain an explicit expression for this quantity in terms of the pairwise state overlaps. Specifically, we show that the deviation from the standard Welch bounds quantifies the average error. This allows us to meaningfully compare different ensembles of fixed cardinality.
Combining these tools, we prove general lower bounds on how well ensembles of a given cardinality can approximate a $k$-design, both in an average and a worst-case sense. 
{The method underlying the first bound also extends to projective-toric and unitary designs \cite{Mele2024}, whereas the second bound uses moment relations specific to state designs.}
As a notable application, we show that whenever complete sets of MUBs exist, the corresponding $N=d(d+1)$ states form an optimal approximate $3$-design among all {exact $1$-designs} of that cardinality. A similar result is derived for SICs. {We establish a connection with shadow tomography, where a set of measurements is used to estimate properties of a quantum state. For measurements that form a projective 2-design, the third moment of the ensemble governs the fluctuations of the standard single-shot estimators, as captured by the shadow norm \cite{Huang2020Predicting}. This connects our geometric measure to the accuracy of observable estimation.} Finally, a key technical ingredient in our proofs is an explicit computation of the spectrum and eigenvectors of the partially transposed projector onto the symmetric subspace. Beyond its role here, this result may find applications in representation theory and quantum information, similarly to related constructions for the orthogonal group $SO(d)$~\cite{Nemoz2025}.

{The paper is organized as follows. Section~\ref{sec:prelim} establishes notations and reviews Welch bounds and $k$-designs in Section~\ref{sec:prelim}.  Section~\ref{sec:spec} presents the spectrum of the partially transposed symmetric subspace projector. We connect Welch bounds to approximate designs in Section~\ref{sec:approximate-designs}, followed by the stronger Welch bounds in Section~\ref{sec:stronger-welch}. Sections~\ref{Sec:MUBSIC} and~\ref{sec:shadow-tomography} discuss applications to SICs, MUBs and shadow tomography. Finally, we provide a conclusion and outlook in Section \ref{sec:conclusion}.}

\section{Preliminaries}\label{sec:prelim}

We fix notation (see Appendix~\ref{app:symbols} for a table of symbols) and recall the essential ingredients used throughout: frame potentials and Welch bounds, complex projective $k$-designs and their operator characterization, and the partially transposed moment operators that will underlie our strengthened bounds.

We fix a dimension $d\ge 2$ and consider a finite set (frame) of unit vectors $\chi:=\{\ket{\psi_i}\in\C^d\}_{i=1}^N$.
Our main quantitative object is the $k$-frame potential, i.e.\ the $2k$-th moment of pairwise overlaps:
\begin{equation}\label{Eq:kEnergyDef}
  \mathcal{E}_k(\chi)
  := \frac{1}{N^2}\sum_{i,j=1}^N \abs{\braket{\psi_i}{\psi_j}}^{2k},
  \qquad k\in \mathbb{N}^+.
\end{equation}
We now recall the well known lower bound on this quantity.

\subsection{Welch bounds}
The Welch bounds are a family of inequalities that constrain the moments of the overlap distribution of any frame.

\begin{thm}[Welch bounds~\cite{WelchBound}]\label{Th:Welch}
For any frame $\chi\subset\C^d$ and any integer $k\ge 1$, one has that
\begin{equation}\label{Eq:WelchBound}
  \mathcal{E}_k(\chi)\ \ge\ \binom{d+k-1}{k}^{-1}.
\end{equation}
\end{thm}

Two basic features of~\eqref{Eq:WelchBound} are worth emphasizing. First, for fixed $d$ and $N$ the right-hand side
decays rapidly with $k$, hence the inequality becomes non-informative at large order.
Second, equality at order $k$ can only occur if $N$ is sufficiently large.
Let $N_{\min}(k,d)$ denote the smallest cardinality for which the $k$-th Welch bound can be saturated in $\C^d$.
It is known that $N_{\min}(k,d)$ is finite~\cite{Existence-k-design(1984)} and satisfies the lower bound
\cite{Levenshtein1992,Dunkl1979DiscreteQA,SmallestT-design(2007)}
\begin{equation}\label{Eq:NminLowerBound}
  \small N_{\min}(k,d)\ \ge\
  \binom{d+\lfloor k/2\rfloor-1}{\lfloor k/2\rfloor}\,
  \binom{d+\lceil k/2\rceil-1}{\lceil k/2\rceil}
  \;=:\; D(k,d).
\end{equation}
A set achieving this lower bound for given $(k,d)$ is called a tight $k$-design. Known existence results are very restrictive: in dimension $d=2$ for
$k=1,2,3,5$~\cite{ThightFramesExistence(1996)}, while for $d\ge3$ they can only exist
for $k=1,2,3$~\cite{BannaiHoggar1985,Hoggar1989,BannaiHoggar(1989)}. $2$-designs
are conjectured to exist in all dimensions~\cite{renes2004symmetric,Zauner1999},
whereas explicit constructions of $3$-designs are known only in some even
dimensions~\cite{hoggar1982t}.
A known consequence of the Welch bounds is the Gerzon bound~\cite{GerzonBound(1974)}:
there are at most $d^2$ equiangular lines in $\C^d$.

\subsection{Complex projective $k$-designs}

The ensembles that saturate all Welch bounds up to order $k$ are precisely the
complex projective $k$-designs.
In words, $k$-designs are frames such that the average of any $k$-degree polynomial over the Haar measure can be computed as a discrete average over the frame.

More formally, we write the Haar average of a function $f$, and the ensemble averages of $f$ as follows:
\[
\underset{\phi\sim H}{\mathbb{E}}[f(\phi)]
:=\int_{\mathbb{P}\C^{d}}f(\phi)\,d\phi,
\;\; 
\underset{\phi\sim \chi}{\mathbb{E}}[\,f(\phi)\,]
:=\frac{1}{N}\sum_{i=1}^{N}f(\psi_i),
\]
where $\chi=\{\ket{\psi_i}\}_{i=1}^N\subset\C^{d}$ is a discrete frame.
We denote the set of bi-homogeneous polynomials of degree
$(k,k)$ as $\mathbf{P}_{k}$. Equivalently, the complex vector space of such polynomials $\mathbf{P}_{k}$ is isomorphic to the set of linear operators acting on the totally symmetric subspace \footnote{Let $\mathcal H_d \cong \mathbb C^d$. We denote by
$\vee^k\C^d := \mathrm{Sym}^k(\mathcal H)
= \{\,|\phi\rangle \in \mathcal H_d^{\otimes k} : R_\pi |\phi\rangle = |\phi\rangle,\ \forall \pi\in S_k\,\}$
the totally symmetric subspace of $\mathcal{H}_d^{\otimes k}$, where $R_\pi$ is the unitary permuting tensor factors according to permutation $\pi$. Note that $\dim(\vee^k\C^d)=\binom{d+k-1}{k}$.}, $\vee^{k}\C^d\subset(\C^d)^{\otimes k}$, which we denote by $\mathcal{L}(\vee^{k}\C^d)$.

The mapping is simply given by
\begin{equation}\label{Eq:PolEvaluation}
p(\phi)=\bra{\phi}^{\otimes k} M_p \ket{\phi}^{\otimes k},
\end{equation}
for a unique $M_p\in\mathcal{L}[\vee^k \mathbb{C}^{d}]$. 

With these notations in place we can give the following definition:
\begin{defn}[Complex projective $k$-design]\label{Def:kDesign}
A frame $\chi$ is a (uniform) complex projective $k$-design if $\forall \;  p \in \mathbf{P}_{k}$
\begin{equation}
  \underset{\phi\sim H}{\mathbb{E}}[p(\phi)]
    =\underset{\phi\sim \chi}{\mathbb{E}}[p(\phi)]
\end{equation}
\end{defn}

A convenient way to formulate saturation of the \(k\)-th Welch bound is through the \(k\)-th order frame operator.  
For a frame \(\chi=\{\ket{\psi_i}\}_{i=1}^N\subset\C^{d}\), define
\begin{equation}
  \mathcal F_k(\chi)
  :=\frac{1}{N}\sum_{i=1}^N \ketbra{\psi_i}^{\otimes k}.
\end{equation}
We can now state the standard equivalence linking \(k\)-designs, saturation of the \(k\)-th Welch bound, and the frame operator \(\mathcal F_k(\chi)\).
\begin{thm}[Welch saturation and designs~\cite{WelchBoundGeom(2012),klappenecker2005,Roy2007}]
\label{Th:WelchDesignEquiv}
For any frame $\chi\subset\C^d$, the following are equivalent:
\begin{enumerate}
\item $\chi$ saturates the $t$-th Welch bound for every $t\le k$;
\item $\chi$ is a complex projective $k$-design in the sense of
Definition~\ref{Def:kDesign}.
\item $\mathcal{F}_{k}(\chi)=\frac{\Pi_k}{\Tr(\Pi_k)}:=\rho_{k}$, where $\Pi_k=\frac{1}{k!}\sum_{\pi\in S_k} R_\pi$ is the projector onto the symmetric subspace $\vee^{k}\C^{d}$. 
\end{enumerate}
\end{thm}

The equivalence between items~(2) and~(3) follows from the well-known identity~\cite{harrow2013churchsymmetricsubspace}
\begin{equation}\label{Eq:HaarMomentEqualsRho}
   \int_{\mathbb{P}\C^{d}} \ketbra{\phi}^{\otimes k}\,d\phi=\rho_{k},
\end{equation}
together with the identification of $\mathbf P_k$ with $\mathcal L(\vee^k\C^d)$ via
$p(\phi)=\bra{\phi}^{\otimes k}M_p\ket{\phi}^{\otimes k}$.

It is also instructive to re-derive the lower bound on the cardinality of a $k$-design $N\ge D(k,d)$ using the same operator viewpoint.
If $\chi$ is a $k$-design then $\mathcal F_k(\chi)=\rho_k$, and applying partial transposition on the last
$\lceil k/2\rceil$ subsystems gives
\[
\mathcal F_k^{\mathsf{\Gamma}}(\chi)=\rho_k^{\mathsf{\Gamma}}.
\]
While partial transposition does not increase the rank of $\mathcal F_k(\chi)$ (it remains a sum of $N$ rank-one operators),
the operator $\rho_k^{\mathsf{\Gamma}}$ typically has much larger rank: in fact
$\operatorname{rank}(\rho_k^{\mathsf{\Gamma}})=D(k,d)$.
Hence equality forces $N\ge D(k,d)$.

This simple argument already clarifies why partial transposition is the natural perspective for what follows. After partial transposition,
$\mathcal F_k^{\mathsf{\Gamma}}(\chi)$ remains a sum of $N$ rank-one terms, hence
$\operatorname{rank}(\mathcal F_k^{\mathsf{\Gamma}}(\chi))\le N$, while
$\rho_k^{\mathsf{\Gamma}}$ has rank $D(k,d)$ and a nontrivial spectrum. Therefore, for
$N<D(k,d)$ an exact match is impossible, and the unavoidable mismatch is controlled by the spectral data of
$\rho_k^{\mathsf{\Gamma}}$. This is precisely the mechanism behind our strengthened Welch bounds, so we next
state the spectrum of $\rho_k^{\mathsf{\Gamma}}$ in full generality.

\section{Spectrum of partially transposed symmetric subspace projector} \label{sec:spec}
We now determine the spectrum of the partially transposed Haar moment operator, including all multiplicities.

For integers $n,m\ge 0$ we define
\begin{equation}\label{Eq:rho_nm_def}
  \rho_{n,m}
  := \int_{\mathbb{P}\C^{d}} \ketbra{\phi}^{\otimes n} \otimes \ketbra{\bar{\phi}}^{\otimes m}\,d\phi \,.
\end{equation}
Equivalently, \(\rho_{n,m}\) is obtained by partially transposing (on the last \(m\) subsystems) the normalized projector onto the symmetric subspace of \(n+m\) copies.
To proceed, we now determine its spectrum (eigenvalues and multiplicities), which will be the key input in the derivation of the strengthened Welch bounds.

\begin{thm}[Spectrum of $\rho_{n,m}$]\label{Th:Spec_rho_nm}
Let $d\ge 2$ and $n,m\ge 0$. The operator $\rho_{n,m}$ can be written as
\begin{equation}
\label{Eq:rho_n,m_spectrum}
\rho_{n,m} = \sumab{r=0}{\min(n,m)} C_{r}(n,m,d)\Pi^{(r)},
\end{equation}
where $\Pi^{(r)}$ are pairwise orthogonal projectors (see Appendix~\ref{app:spectrum} for the details) with eigenvalues
\begin{equation} \label{Eq:EigenvaluesFormula}
C_r(n,m,d) = \frac{\binom{n+m+d-1}{r}}{\binom{n+m}{m}\binom{n+m+d-1}{n+m}}.
\end{equation}
Moreover, for every $r\in\{0,1,\dotsc,\min(n,m)\}$ projectors $\Pi^{(r)}$ have ranks $\eta_r(n,m,d)$, given by
\begin{equation} \label{eq:eta}
\eta_r(n,m,d) = \tbinom{n-r+d-2}{d-2} \tbinom{m-r+d-2}{d-2} \tfrac{n+m-2r+d-1}{d-1}.
\end{equation} 
\end{thm}

The proof of Theorem~\ref{Th:Spec_rho_nm} is given in Appendix~\ref{app:spectrum}. The eigenspaces of $\rho_{n,m}$, i.e. $V_{r}:= \{ \ket{\phi}: \rho_{n,m}\ket{\phi}= C_r(n,m,d)\ket{\phi}\}$, are closely connected to the representation theory of $U(d)$, and to the mixed Schur--Weyl duality. Indeed, they are the isotypic subspaces of the action $\mathcal{U}_{n,m}^{d}:=\{U^{\otimes n}\otimes \bar{U}^{m}:U\in U(d)\}$. We give an explicit description of $V_{r}$ for every $r\in\{0,\dotsc,\min(n,m)\}$ in Appendix~\ref{app:spectrum}.

Equation \eqref{Eq:rho_n,m_spectrum} and the characterization of $V_{r}$ are the key quantitative ingredients of our new bounds. 
They convert the partially transposed rank-constrained approximation problem into an explicit eigenvalue optimization, which yields the improvement terms over standard Welch bounds. 
Theorem~\ref{Th:StrongerWB} follows from this spectral input for arbitrary frames, and Theorem~\ref{Th:VS-WB} refines it under additional lower-order design constraints. 
The constants and thresholds in both theorems are determined directly by the spectrum of $\rho_k^\mathsf{\Gamma} =\rho_{\lceil k/2\rceil,\lfloor k/2\rfloor}$ .

For context, our result is the unitary counterpart of the analogous spectral decomposition results obtained earlier for the orthogonal group \cite{schatzki2024orthogonal,nemoz2025orthogonal} and the symplectic group \cite{west2024symplectic}. {Our spectrum also confirms the conjectured formula for the smallest nonzero eigenvalue used in Proposition~2 of Ref.~\cite{abellanet2025sufficient} to derive criteria for symmetric absolutely PPT states.}

\section{Approximate designs}\label{sec:approximate-designs}

As discussed in the preliminaries, exact complex projective $k$-designs can only exist when the
cardinality $N$ is sufficiently large. It is therefore natural to ask how well a set of
$N<N_{\min}(k,d)$ states can \emph{approximate} a $k$-design. In this section we recall a
standard worst-case notion of approximation and introduce a complementary average-case
figure of merit that admits a closed form.

\subsection{Worst-case error}

A natural figure of merit is the largest deviation that can occur when replacing the Haar average
by the discrete average over $\chi$:
\begin{equation}\label{Eq:epsMaxDef}
  \epsilon_{\rm{max}}^{(k)}(\chi):=
  \underset{p\in \mathbf{P}_{k}}{\max}\;
  \lVert M_{p}\rVert_{{1}}^{-1}
  \abs{\;
    \underset{\phi\sim H}{\mathbb{E}}[p(\phi)]
    -\underset{\phi\sim \chi}{\mathbb{E}}[p(\phi)]
  \;}.
\end{equation}
The factor $\lVert M_{p}\rVert_{{1}}^{-1}$ normalizes the error, equivalently the maximization could be carried out over polynomials {whose associated operators have unit trace norm}.

Using~\eqref{Eq:PolEvaluation} and the definitions of $\mathcal F_k(\chi)$ and $\rho_k$ from the preliminaries,
\begin{align}
\underset{\phi\sim H}{\mathbb{E}}[p(\phi)]
-\underset{\phi\sim\chi}{\mathbb{E}}[p(\phi)]
&=
\Tr\!\left[M_p\left(\rho_k-\mathcal F_k(\chi)\right)\right].
\end{align}
It follows that
\begin{equation}\label{Eq:epsMaxOpNorm}
  \epsilon_{\rm{max}}^{(k)}(\chi)=\lVert \mathcal{F}_{k}(\chi)-\rho_{k}\rVert_{\infty}.
\end{equation}
Indeed, the maximization in~\eqref{Eq:epsMaxDef} is saturated by $M_p=\ketbra{\varphi}$, where
$\ket{\varphi}\in\vee^k\C^d$ is an eigenvector {of $\mathcal F_k(\chi)-\rho_k$ corresponding to an eigenvalue of maximal absolute value}.

\subsection{Average error}

To quantify instead a typical deviation, we equip $\mathcal L[\vee^k\C^d]$ with  a probability measure $\mu$
and define the average squared error
\begin{equation}\label{Eq:epsAvgDef}
  \epsilon_{\rm{avg}}^{(k)}(\chi):=
  \sqrt{
  \underset{M_{p}\sim \mu}{\mathbb{E}}
  \left[
    \lVert M_{p}\rVert_{2}^{-2}
    \abs{
      \underset{\phi\sim H}{\mathbb{E}}[p(\phi)]
      -\underset{\phi\sim \chi}{\mathbb{E}}[p(\phi)]
    }^{2}
  \right]}.
\end{equation}
For concreteness, we fix $\mu$ to be the {complex} Ginibre ensemble~\cite{Ginibre1965}{~\cite{byun2023progressstudyginibreensembles}} on $\mathcal L[\vee^k\C^d]$, {i.e.\ matrices with independent standard complex Gaussian entries}; see Appendix~\ref{Sec:Eq(18)proof} for the normalization. (We note that {what follows holds for other ensembles as well}.)
We now prove an explicit formula for $\epsilon_{\rm{avg}}^{(k)}(\chi)$.

\begin{thm}\label{Th:AvgErrorExplicit}
Let $\chi$ be a frame in $\mathbb{C}^{d}$ . Then
\begin{equation}\label{Eq:AvgErrorEnergy}
   \epsilon_{\rm{avg}}^{(k)}(\chi)=
   \binom{d+k-1}{k}^{{-1}}
   \left(\mathcal{E}_{k}(\chi)-\binom{d+k-1}{k}^{-1}\right)^{\frac{1}{2}}.
\end{equation}
Equivalently,
\begin{equation}\label{Eq:AvgErrorFrob}
   \epsilon_{\rm{avg}}^{(k)}(\chi)=
   \binom{d+k-1}{k}^{{-1}}
   \lVert \mathcal{F}_{k}(\chi)-\rho_{k} \rVert_{2}.
\end{equation}
\end{thm}

\begin{proof}
Let $D=\dim(\vee^k\C^d)=\binom{d+k-1}{k}$. {The operator space $\mathcal L[\vee^k\C^d]$ has complex dimension $D^2$.}
For the Ginibre ensemble $\mu$ on $D\times D$ matrices one has
\begin{equation}\label{Eq:GUEIdentity}
  \underset{M\sim \mu}{\mathbb{E}}
  \left[
    \abs{\Tr\!\left[\frac{M}{\lVert M\rVert_{2}}A\right]}^{2}
  \right]
  =\frac{\lVert A\rVert_{2}^{2}}{{D^2}}
\end{equation}
for every self-adjoint $A\in\mathcal L[\vee^k\C^d]$ (see Appendix~\ref{Sec:Eq(18)proof} for details).
Applying~\eqref{Eq:GUEIdentity} to $A=\rho_k-\mathcal F_k(\chi)$ gives
\[
\epsilon_{\rm{avg}}^{(k)}(\chi)^2
=\binom{d+k-1}{k}^{{-2}}\,\lVert \mathcal F_k(\chi)-\rho_k\rVert_2^2.
\]

Finally, using $\Tr[\mathcal F_k(\chi)]=\Tr[\rho_k]=1$ one checks that
\[
\lVert \mathcal F_k(\chi)-\rho_k\rVert_2^2
=\lVert \mathcal F_k(\chi)\rVert_2^2-\lVert \rho_k\rVert_2^2.
\]
Moreover, $\lVert\mathcal F_k(\chi)\rVert_2^2=\mathcal E_k(\chi)$ and
$\lVert\rho_k\rVert_2^2=\binom{d+k-1}{k}^{-1}$, which yields~\eqref{Eq:AvgErrorEnergy}.
\end{proof}

{In the language above, $\epsilon_{\rm avg}^{(k)}(\chi)$ is the root-mean-square error incurred when the Haar integral is replaced by the average over the $N$ states in $\chi$, with the average taken isotropically over normalized polynomials in $\mathbf P_k$.}

\medskip

In some applications one is interested in a \emph{distinguishability} notion of approximation.
In cryptographic and pseudorandomness settings \cite{Brakerski2019}, one asks how well an ensemble can be distinguished from Haar given \(k\) copies. This is naturally quantified by the trace norm
\(\|\mathcal F_k(\chi)-\rho_k\|_1\), since it upper bounds the optimal distinguishing advantage of any measurement on \(k\) copies \cite{Helstrom1969}.

Let $A:=\mathcal F_k(\chi)-\rho_k\in\mathcal L(\vee^k\C^d)$.
Standard norm inequalities imply
\[
\|A\|_{\infty}\ge \frac{\|A\|_2}{\sqrt{{\binom{d+k-1}{k}}}}.
\]
Moreover, since \(A\) is Hermitian and traceless, writing its positive/negative parts gives
\[
\|A\|_1 \ge \sqrt{2}\,\|A\|_2.
\]
Therefore, any lower bound on \(\|A\|_2\) immediately yields lower bounds in both operator norm and trace norm\footnote{{The coefficient $\sqrt{2}$ is in general tight, for example equality is reached by the ensemble $\chi:\{ \ket{\phi}, \ket{\psi_{2}}\dots..\ket{\psi_{N}}\}$ where $\ket{\phi}\neq \ket{\psi_{1}}$ and $(\chi/\{\ket{\phi}\}) \cup \{\ket{\psi_{1}}\}$ is a $k-$design. Similar arguments show that the $\infty$-norm inequality above is tight.} }.

Summarizing, \(\epsilon_{\rm max}^{(k)}(\chi)\) and \(\epsilon_{\rm avg}^{(k)}(\chi)\) quantify worst-case and average-case deviations from Haar moments. For us \(\epsilon_{\rm avg}^{(k)}\) is the most natural figure of merit: by Theorem~\ref{Th:AvgErrorExplicit}, {its square} is directly proportional to the excess \(k\)-frame potential above the \(k\)-th Welch value. This gives strong motivation for strengthened Welch bounds, as they directly quantify the minimum average error achievable at fixed cardinality.

\section{Stronger Welch bounds}\label{sec:stronger-welch}

We now address the central regime of this work: finite frames that are too small to realize an exact $k$-design.
In this regime, the standard $k$-th Welch bound is non-saturable and often quantitatively weak.
Our goal is therefore to derive sharpened lower bounds on the $k$-frame potential
$\mathcal E_k(\chi)$ that remain informative below the design threshold.

Equivalently, these bounds control the unavoidable discrepancy
$\|\mathcal F_k(\chi)-\rho_k\|_2$, and hence the best achievable accuracy with which a frame
can approximate a $k$-design, measured by $\epsilon_{\mathrm{avg}}^{(k)}(\chi)$.
As a by-product, our method also yields a lower bound on the worst-case error
$\epsilon_{\mathrm{max}}^{(k)}(\chi)$.

\subsection{General bounds}

We can now state and prove our first strengthened form of the Welch bounds.
Let $\mathsf{\Gamma}$ denote partial transposition on the last $\lceil k/2\rceil$ tensor factors. 

\begin{thm}[Stronger Welch bounds]\label{Th:StrongerWB}
For any frame $\chi=\{\ket{\psi_i}\}_{i=1}^N\subset\C^d$ and any $k\ge 1$,
\begin{align}
\|\mathcal F_k(\chi)-\rho_k\|_2^2 &\;\ge\; \Delta_2 + \frac{\Delta_1^2}{N}, \label{Eq:ApproximateDesignIneq}\\
\mathcal E_k(\chi) &\;\ge\; \binom{d+k-1}{k}^{-1} + \Delta_2 + \frac{\Delta_1^2}{N}, \label{Eq:StrongherWB}
\end{align}
where $\Delta_\ell:=\sum_{i=N+1}^{D(k,d)}\lambda_i^\ell$ and
$\lambda_1\ge\cdots\ge \lambda_{D(k,d)}>0$ are the eigenvalues of $\rho_k^{\mathsf{\Gamma}}$.
\end{thm}
The complete spectrum of $\rho_k^{\Gamma}$ required to compute $\Delta_\ell$ is implemented in our GitHub repository \cite{Code}.

\begin{proof}
We first reduce to a rank--constrained approximation problem.
Since the Frobenius norm is invariant under partial transposition,
\[
\|\mathcal F_k(\chi)-\rho_k\|_2=\|\mathcal F_{k}^{\mathsf{\Gamma}}(\chi)-\rho_k^{\mathsf{\Gamma}}\|_2.
\]
Moreover,
\[
\mathcal F_{k}^{\mathsf{\Gamma}}(\chi)
=\frac{1}{N}\sum_{i=1}^N
\ketbra{\psi_i}^{\otimes m}\otimes \ketbra{\bar\psi_i}^{\otimes n}
\]
is a sum of $N$ rank-one operators, hence $\operatorname{rank}(\mathcal F_{k}^{\mathsf{\Gamma}}(\chi))\le N$, and
$\Tr[\mathcal F_{k}^{\mathsf{\Gamma}}(\chi)]=\Tr[\rho_k^{\mathsf{\Gamma}}]=1$.

{If $N\ge D(k,d)$, the claimed inequality is immediate, so we assume $N<D(k,d)$.}
At this point we use a simple spectral minimization principle: among all positive semidefinite operators of rank at
most $N$ and fixed trace, the one closest in Frobenius norm to a given positive operator $Y$ is obtained by
aligning eigenbases and the kernel of $X$ corresponding to the lowest eigenvalues of $Y$. Quantitatively, if $Y\ge 0$ has
eigenvalues $\lambda_1\ge\cdots\ge\lambda_m>0$ and $X\ge0$ satisfies $\Tr[X]=\Tr[Y]$ and $\operatorname{rank}(X)\le N<m$, then
\begin{equation}\label{Eq:RankTraceFrob_inProof}
\|X-Y\|_2^2 \;\ge\; \Delta_2 + \frac{\Delta_1^2}{N},
\qquad
\Delta_\ell := \sum_{i={N+1}}^m {\lambda_i^\ell} .
\end{equation}
{Equality holds when $X$ and $Y$ commute and the eigenvalues of $X$ are
$x_i=\lambda_i+\Delta_1/N$ for $i\le N$ and $x_i=0$ for $i>N$.}
(An independent proof is given in Appendix~\ref{app:RankTraceLemma}, but this fact can also be seen as a consequence of the Hoffman--Wielandt inequality \cite{HoffmanWielandt1953}.)

We now apply this to $X=\mathcal F_{k}^{\mathsf{\Gamma}}(\chi)$ and $Y=\rho_k^{\mathsf{\Gamma}}$,
which immediately gives~\eqref{Eq:ApproximateDesignIneq}.

Finally, using $\|\mathcal F_k(\chi)-\rho_k\|_2^2=\|\mathcal F_k(\chi)\|_2^2-\|\rho_k\|_2^2$ together with
$\|\mathcal F_k(\chi)\|_2^2=\mathcal E_k(\chi)$ and $\|\rho_k\|_2^2=\binom{d+k-1}{k}^{-1}$, we obtain
\eqref{Eq:StrongherWB}.
\end{proof}
{\paragraph*{Remark: Choice of partial transpose.}
Transposing any $m\leq \lfloor k/2\rfloor$ of the $k$ copies gives a valid bound. Let
\[
h_j:=\binom{d+j-1}{j},
\qquad q:=\lfloor k/2\rfloor .
\]
The support rank associated with the split $(m,k-m)$ is $h_mh_{k-m}$ and is maximal for $m=q$. Consequently, whenever
\[
N\geq h_{q-1}h_{k-q+1}
=
\binom{d+q-2}{q-1}
\binom{d+k-q}{k-q+1},
\]
no unbalanced choice improves the standard Welch bound, whereas the balanced choice remains nontrivial whenever $N<D(k,d)$. For  $N<h_{q-1}h_{k-q+1}$, an unbalanced partial transpose can sometimes give a stronger bound, in which case one may want to apply theorem \ref{Th:StrongerWB} for a different $m$.}

We illustrate our stronger bounds and their comparison to the standard Welch bounds for $k=2,3$ for $d=2$ for a range of different $N$ in Fig.~\ref{fig:two_plots_side_by_side} (left panel). We also plot the results of a heuristic search, which indicates that these bounds remain fairly tight even for $N$ far from the design threshold, where the usual Welch bounds become loose or uninformative. These were obtained by directly minimizing the $k$-frame potential $\mathcal{E}_k(\chi)$
over $N$ unit vectors in $\mathbb{C}^d$ using a multi-start non-convex optimization. Concretely, each state is parametrized by $2d-2$ hypersphere coordinates (amplitude angles and relative phases, with one global phase fixed), and the optimization is performed from multiple random initializations with local refinement.

Theorem~\ref{Th:StrongerWB} shows that strengthening the Welch bounds below the design threshold reduces to
understanding the spectrum of the partially transposed Haar moment operator $\rho_k^{\mathsf{\Gamma}}$.
In particular, the correction terms $\Delta_1$ and $\Delta_2$ are explicit functions of its ordered eigenvalues.

\subsection{Bounds for $k'<k$ designs}

We now refine the universal bound by exploiting additional moment structure.
A central intermediate regime is when $N$ is sufficient to realize an exact $k'$-design for some $k'<k$,
but still insufficient for an exact $k$-design.
In this case, we are interested in optimal approximate $k$-design that are exact $k'$-design.
%admissible frames satisfy exact constraints up to order $k'$, so the feasible set is much smaller than in Theorem~\ref{Th:StrongerWB}.
As the feasible set is much smaller than in Theorem~\ref{Th:StrongerWB},
one expects strictly stronger lower bounds than the universal estimate.

At a high level, we work in a subspace of the mixed Schur--Weyl decomposition of $(\mathbb{C}^d)^{\otimes n} \otimes ((\mathbb{C}^d)^{\otimes m})^*$:
\begin{equation}
    \vee^n \mathbb{C}^d \otimes (\vee^m \mathbb{C}^d)^* 
    \simeq \bigoplus_{r=0}^{\min(n,m)} V_r
\end{equation}
in which $\rho_k^{\mathsf{\Gamma}}$ is block diagonal and acts as a scalar on each isotypic block $V_r$.
For a $k'$-design, we show that many blocks of $\mathcal F_k^{\mathsf{\Gamma}}(\chi)$ are already fixed to their Haar values (and certain off-diagonal blocks are forced to vanish), so only a restricted set of low-$r$ blocks can deviate.
The problem therefore becomes a constrained low-rank approximation in Frobenius norm over the remaining free blocks. Figure~\ref{fig:block-k4-kp2} summarizes which blocks are fixed by the $k'$-design constraints and which remain free; the formal statement is given in Proposition~\ref{Prop:BlockDiagonalFrameOp} (Appendix~\ref{app:Welchproof}).
\begin{figure}[h!]
\centering
\renewcommand{\arraystretch}{1.35}

\definecolor{freeblock}{RGB}{255,235,235}
\definecolor{fixedblock}{RGB}{235,245,255}
\definecolor{zeroblock}{RGB}{245,245,245}

\[
\mathcal F_{2,2}^\mathsf{\Gamma}
=
\left[
\begin{array}{c|ccc}
 & r'=0 & r'=1 & r'=2\\
\hline
r=0 &
\cellcolor{freeblock} M_{00} &
\cellcolor{freeblock} M_{01} &
\cellcolor{zeroblock} \bm{0}
\\
r=1 &
\cellcolor{freeblock} M_{10}=M_{01}^{\dagger} &
\cellcolor{fixedblock} C_{1}\Pi^{(1)} &
\cellcolor{zeroblock} \bm{0}
\\
r=2 &
\cellcolor{zeroblock} \bm{0} &
\cellcolor{zeroblock} \bm{0} &
\cellcolor{fixedblock} C_{2}\Pi^{(2)}
\end{array}
\right].
\]

\vspace{1mm}
\[
\begin{aligned}
&\text{Free blocks (red): }(0,0),(0,1),(1,0),\\
&\text{Fixed Haar blocks (blue): }(1,1),(2,2),\\
&\text{Forced zero (gray): }r\neq r'\ \text{with}\ r+r'\ge 2.
\end{aligned}
\]

\vspace{1mm}
\[
\dim V_r=\eta_r(4,d),\qquad
M_{rr'}\in\mathrm{Hom}(V_{r'},V_r)\cong \mathbb C^{\eta_r\times\eta_{r'}}.
\]

\caption{Block decomposition of \(\mathcal F_{2,2}\) for \(k=4,\ k'=2\). 
The \(k'\)-design constraints force agreement with Haar on all blocks with \(r+r'\ge k-k'=2\), leaving only the red blocks as free variables.}
\label{fig:block-k4-kp2}
\end{figure}

Using positivity, rank, and fixed block-trace constraints, we show that the optimum is obtained by placing the rank deficit in the block with smallest eigenvalue of $\rho_k^{\mathsf{\Gamma}}$, which yields the explicit improvement below.

\begin{thm}[Stronger Welch bound for designs]\label{Th:VS-WB}
Let $\chi$ be a $(k-2)$-design with
$N\ge D(k-2,d)+\frac{\eta_{0}(k,d)}{k+d-1}$. Then
\begin{align} \label{eq:Norm2Th6}
\|\mathcal F_k(\chi)-\rho_k\|_2^2 &\ge \Delta(k,d,N),\\
\mathcal E_k(\chi) &\ge \binom{d+k-1}{k}^{-1} + \Delta(k,d,N),
\end{align}
where 
\begin{align}\label{eq:delta}
\Delta(k,d,N)
:= C_{0}^2(k,d)\left(s+\frac{s^{2}}{\eta_{0}(k,d)-s}\right),
\\ s:=D(k,d)-N.
\end{align}
Here $C_{0}(k,d) := C_{0}\!\left(\lfloor k/2 \rfloor, \lceil k/2 \rceil, d\right)$ 
and $\eta_{0}(k,d) := \eta_{0}\!\left(\lfloor k/2 \rfloor, \lceil k/2 \rceil, d\right)$ are given in Eqs.\eqref{Eq:EigenvaluesFormula} and \eqref{eq:eta} respectively.
\end{thm}

We remark that if $\chi$ is a $k-2$-design, this already implies $\abs{\chi}\geq D(k-2,d)$, so there is just a small range of cardinalities for which the theorem requires an additional assumption on the cardinality. Further, as $D(k-1,d)\geq D(k-2,d)+\frac{\eta_{0}(k,d)}{k+d-1}$, all $k-1$-designs respect the theorem regardless their cardinality.

We refer to Appendix~\ref{app:Welchproof} for the proof, which relies on mixed Schur--Weyl duality and on the
explicit spectral decomposition of $\rho_k^{\mathsf{\Gamma}}$ from \ref{sec:spec}.

{We also make explicit the case where $k=3$, as this concerns frames that constitute a measurement and is thus particularly important.
\begin{corollary}[Stronger Welch bound for 1-designs]\label{Th:VS-WB3}
Let $\chi$ be a $1$-design with
$N\ge \frac{d(d+1)}{2}$. Then
\begin{align}\label{Eq:VS-WB3}
\|\mathcal F_k(\chi)-\rho_k\|_2^2 &\ge \frac{2(d-1)(\frac{d^{2}}{2}(d+1)-N)}
{d(d+1)^{2}(d+2)(N-d)}
\end{align}
\end{corollary}
\begin{proof}
For $k=3$, the quantities appearing in Theorem~6 simplify to
\begin{align}
D(3,d)&=\frac{d^{2}(d+1)}{2},
\\
\eta_{0}(3,d)&=\frac{d(d-1)(d+2)}{2}=D_{3}-d, \\
C_{0}(3,d)&=\frac{2}{d(d+1)(d+2)}.
\end{align}
the theorem is obtained simply by substitution from Eq.~\eqref{eq:Norm2Th6} of theorem \ref{Th:VS-WB}. 
\end{proof}}

{We end this section by noting that the rank-based argument underlying Theorem~\ref{Th:StrongerWB} is not specific to state designs and applies more generally. By contrast, Theorem~\ref{Th:VS-WB} relies on specific spectral properties of the partially trasposed projector on the symmetric subspace. In Appendices~\ref{app:toric-designs} and~\ref{app:unitary-designs}, we apply the same method to projective-toric and unitary designs. In the toric case, it gives a stronger odd-order cardinality bound that answers a question posed in Ref.~\cite{IosueEtAl2024}, as well as a lower bound for undersized designs; in the unitary case, it gives a spectral-tail bound after vectorization.}

\section{MUBs and SICs as optimal approximate 3-designs}\label{Sec:MUBSIC}

\begin{figure*}[t]
    \centering
    \includegraphics[width=0.49\textwidth]{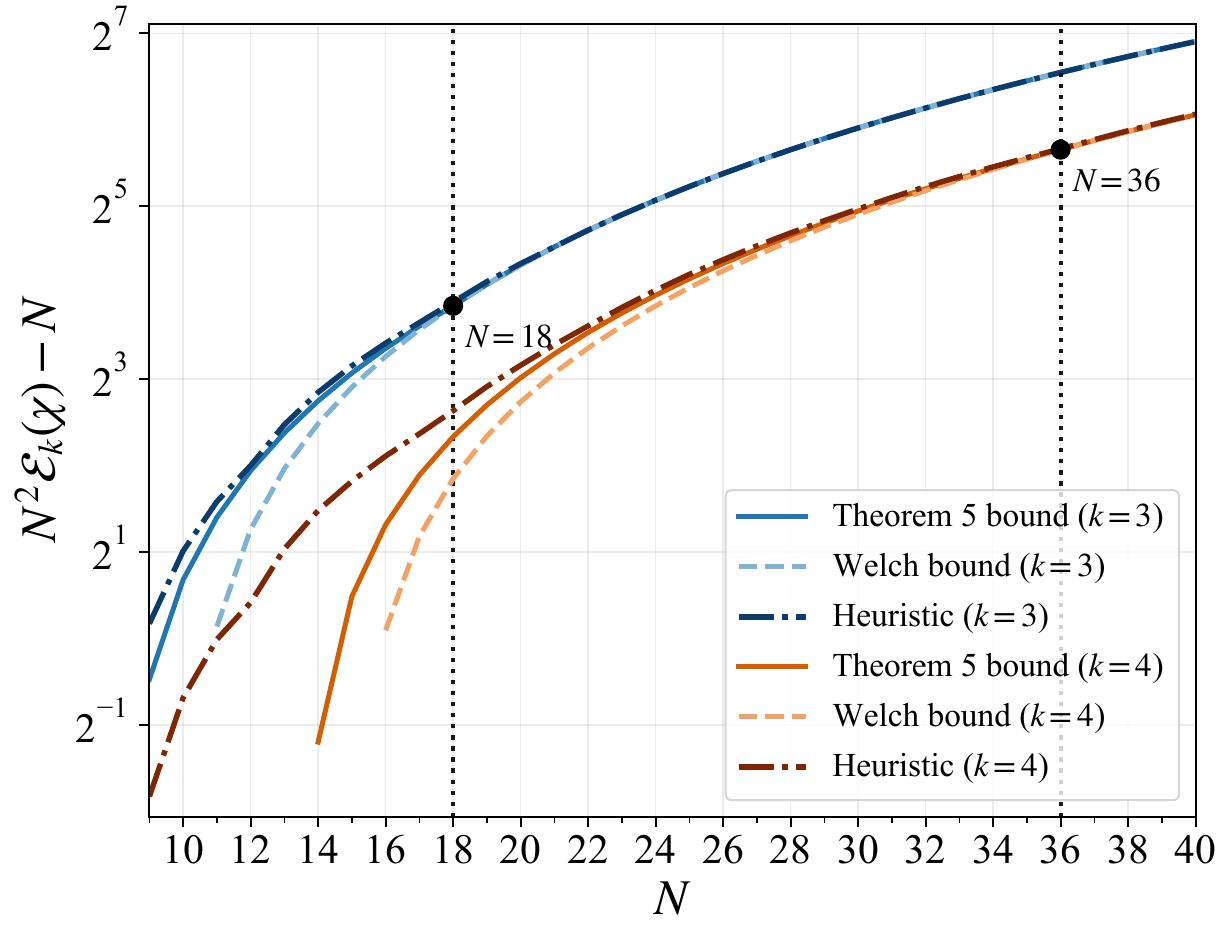}\hfill
    \includegraphics[width=0.49\textwidth]{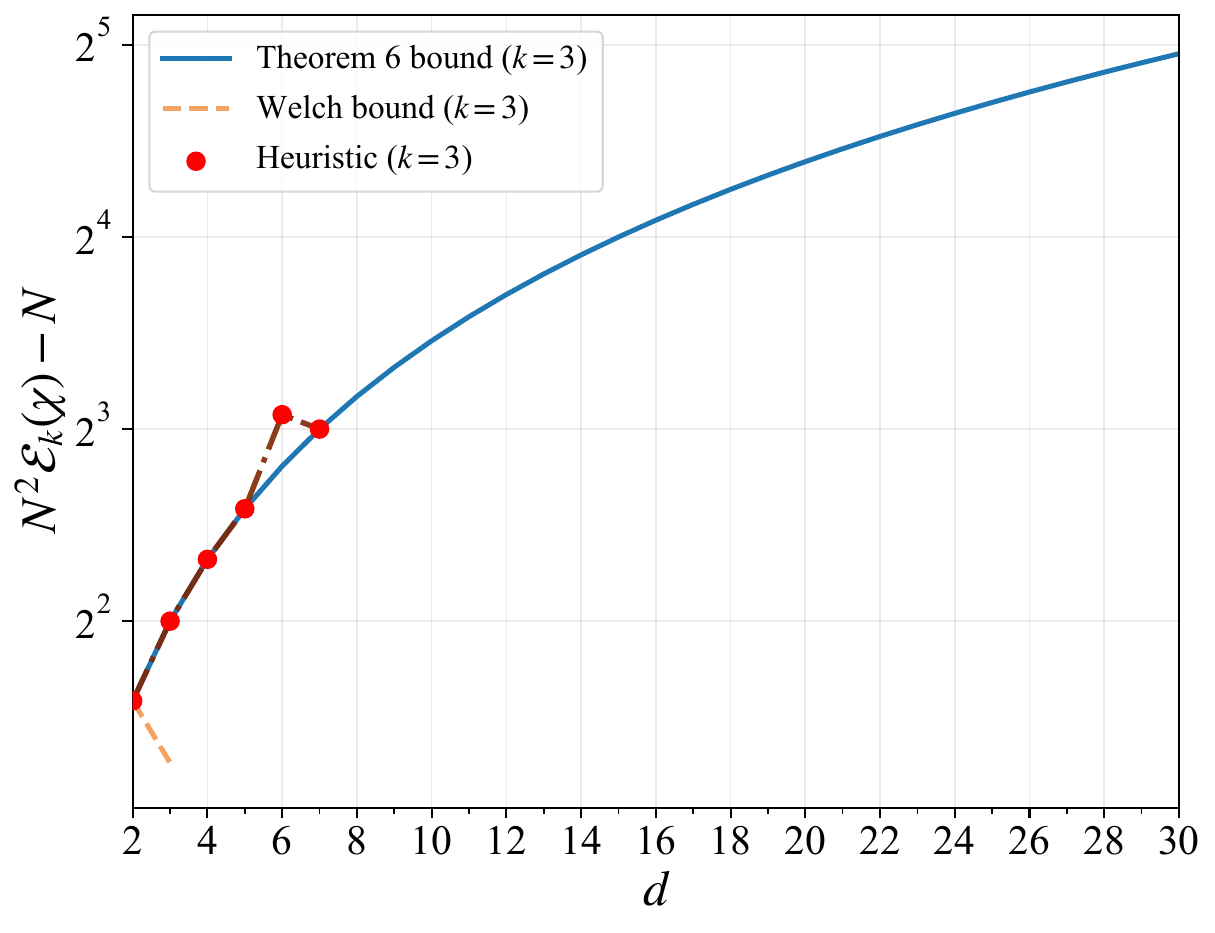}
    \caption{
Lower bounds and {best numerical values found} for the off-diagonal overlap moment ($ \sum_{i \not =j} \bigl|\braket{\psi_i}{\psi_j}\bigr|^{2k}= N^{2}\mathcal E_k(\chi)-N$).\\
\textbf{Left:} $d=3$, $k=3,4$: Welch bound, strengthened bound (Thm.~\ref{Th:StrongerWB}), and {best numerical values found} versus $N$.
Black points/dotted lines mark the design threshold $N_{\min}(k,d)$, where the curves coincide. For $k=3$, the numerical values nearly saturate the strengthened bound.
\textbf{Right:} $k=3$ in the $2$-design regime $N=d(d+1)$: Welch bound, sharpened bound (Thm.~\ref{Th:VS-WB}), and {best numerical values found} versus $d$.
The Welch bound is trivial from $d=4$, while the sharpened bound remains non-trivial. The {best values found are} consistent with the sharpened bound for $d\le 5$ and $d=7,8$, but {show} a clear positive gap (about $20\%$) at $d=6$, adding evidence against the existence of a complete set of MUBs in $d=6$.
The code used to generate the plots is available at~\cite{Code}.
}
    \label{fig:two_plots_side_by_side}
\end{figure*}

We now illustrate the strength of the sharpened Welch bounds by two canonical families of $2$-designs.
Throughout this section we focus on the case $k=3$, i.e.\ we compare exact 1-designs by how well they approximate
a $3$-design, as quantified by $\mathcal E_3(\chi)$ (equivalently $\|\mathcal F_3(\chi)-\rho_3\|_2$ or
$\epsilon^{(3)}_{\rm avg}$).

\subsection{SICs: saturation at $N=d^2$.}
Assume a SIC exists in dimension $d$, i.e.\ a set
$\chi_{\mathrm{SIC}}=\{\ket{\psi_i}\}_{i=1}^{d^2}$ such that
$|\braket{\psi_i}{\psi_j}|^2=\frac{1}{d+1}$ for $i\neq j$.
It is well known that $\chi_{\mathrm{SIC}}$ is a tight complex projective $2$-design~\cite{Renes2004} and thus satisfies the assumptions of theorem \ref{Th:VS-WB}.
A direct computation gives
\begin{equation}\label{eq:E3_SIC}
\mathcal{E}_3(\chi_{\mathrm{SIC}})
= \frac{1}{d^4}\Bigl(d^2 + d^2(d^2-1)\frac{1}{(d+1)^3}\Bigr)
= \frac{d+3}{d(d+1)^2}.
\end{equation}
{The latter matches precisely the right-hand side of Corollary \ref{Th:VS-WB3}}, hence SICs
\emph{saturate} the sharpened inequality at this cardinality. Consequently, SICs minimize
$\|\mathcal F_3(\chi)-\rho_3\|_2$ (equivalently $\epsilon^{(3)}_{\rm avg}$) among all exact {$1$-designs} of size $d^2$.
This is expected, as SICs minimize any convex function of the fidelity matrix, similar to universal minimum-potential sets~\cite{UniversalMinimizSets(2007)}.

\subsection{MUBs: optimal approximate $3$-design of their cardinality}
As a direct corollary of our strengthened bounds{,} already at $k=2$, one recovers the standard upper bound
on the number of mutually unbiased bases: in dimension $d$, there can be at most $d+1$ MUBs \cite{woottersMUB}.
Assume there exists a complete set of $d+1$ mutually unbiased bases (MUBs) in $\C^d$,
$\{\mathcal{B}_\alpha\}_{\alpha=0}^{d}$ with $\mathcal{B}_\alpha=\{\ket{e_{\alpha,i}}\}_{i=1}^d$, such that
\begin{equation}\label{eq:MUB_overlap}
|\braket{e_{\alpha,i}}{e_{\beta,j}}|^2 =
\begin{cases}
\delta_{ij} & \alpha=\beta,\\[2pt]
\frac{1}{d} & \alpha\neq\beta.
\end{cases}
\end{equation}
Let $\chi_{\mathrm{MUB}}:=\bigcup_{\alpha=0}^{d}\mathcal{B}_\alpha$ be the frame consisting of all
$N=d(d+1)$ vectors. It is well known that $\chi_{\mathrm{MUB}}$ is an (unweighted) complex projective $2$-design
\cite{klappenecker2005}. Counting overlaps using~\eqref{eq:MUB_overlap} yields
\begin{equation}\label{eq:E3_MUB}
\mathcal{E}_3(\chi_{\mathrm{MUB}})=\frac{1}{d^2}.
\end{equation}

We now show optimality. {Plugging $N=d(d+1)$ in the right hand side of Eq.\eqref{Eq:VS-WB3} we get the sharp bound $ \|\mathcal F_k(\chi)-\rho_k\|_2^2 \geq \frac{1}{d^2}-\binom{d+2}{3}^{-1}$,
or equivalently 
\begin{equation}\label{Eq:MUB_optimal}
   \mathcal{E}_3(\chi)\ \ge\ \frac{1}{d^2}
 \;\; \forall
 \text{ $1$-design $\chi$ with $|\chi|=d(d+1)$.}
\end{equation} }
% the quantities appearing in~\eqref{eq:delta} simplify to
% \begin{align}\label{Eq:Consts_k3}
% D(3,d)&=\frac{d^2(d+1)}{2},\qquad
% \eta_{0}(3,d)=\frac{d(d-1)(d+2)}{2}, \nonumber \\
% C_{0}(3,d)&=\frac{2}{d(d+1)(d+2)}.
% \end{align}
% Plugging $N=d(d+1)$ into~\eqref{eq:delta} gives
% \[
% \Delta(3,d,d(d+1))
% = \frac{1}{d^2}-\binom{d+2}{3}^{-1},
% \]
% and Theorem~\ref{Th:VS-WB} yields the sharp bound
% \begin{equation}\label{Eq:MUB_optimal}
% \mathcal{E}_3(\chi)\ \ge\ \frac{1}{d^2}
% \;\; \forall
% \text{ $1$-design $\chi$ with $|\chi|=d(d+1)$.}
% \end{equation}
Since $\chi_{\mathrm{MUB}}$ attains~\eqref{eq:E3_MUB}, complete MUBs saturate the sharpened inequality at this
cardinality. Equivalently, they minimize $\|\mathcal F_3(\chi)-\rho_3\|_2$ and thus minimize
$\epsilon^{(3)}_{\rm avg}$ among all exact $1$-designs of size $d(d+1)$.\\

The saturation of~\eqref{Eq:MUB_optimal} suggests a variational route to numerically probing the existence of complete sets of MUBs \cite{Durt2010}. A natural approach is to optimize directly over unions of $d+1$ orthonormal bases. Fixing one basis to the computational basis reduces the search to $d$ independent unitaries, i.e.\ to the manifold $SU(d)^d$. In our numerical method, we parameterize these $d$ bases by complex matrices and project them to unitaries at each optimization step, and we minimize the penalized objective
\begin{equation}
f=\mathcal{E}_3(\chi)+\lambda\!\left(\mathcal{E}_2(\chi)-\frac{2}{d(d+1)}\right)^2,
\end{equation}
Here the first term targets the 3-frame potential, while the penalty enforces approximate $2$-design behaviour. The resulting non-convex optimization is performed with multiple random restarts and local gradient-based refinement.

Applying this procedure in small dimensions, we find that for $d\leq5$, {the best value found agrees with the value attained by complete MUBs} up to machine precision, consistent with the existence of complete MUBs. For \(d=6\), we observe a substantial positive gap, whereas for $d=7$ and $d=8$ we again find no gap within numerical precision. This behavior adds to the existing numerical evidence against the existence of a complete set of MUBs in dimension six \cite{Durt2010,Brierley2008,Grassl2004,Jaming2010,Gribling2024,Colomer2022,McNulty2024-MUBsRev}. The results are illustrated in Fig.~\ref{fig:two_plots_side_by_side}.
It would be interesting to investigate whether such gaps can be
certified rigorously using polynomial relaxation techniques \cite{Parrilo2003}. We note that, contrary to other variational criteria \cite{Durt2010,Colomer2022}, our objective has full permutation symmetry, which should enable significant simplification \cite{Gatermann2004}.

\section{Application to shadow tomography}
\label{sec:shadow-tomography}

{An operational setting in which approximate projective designs arise naturally is classical shadow tomography. In this framework, data obtained from a fixed measurement procedure are used to construct unbiased estimators of the expectation values of several observables~\cite{Huang2020Predicting,Innocenti2023}. An important quantity in this context is the shadow norm: its square is the largest single-shot second moment over all input states, and hence bounds the variance. In this section, we show that the Hilbert--Schmidt distance
$\|\mathcal{F}_{3}(\chi)-\rho_{3}\|_{2}$ studied in this work gives an upper bound on the shadow norm and provides a simple sufficient condition for retaining the dimension-independent variance bound associated with an exact projective $3$-design.

Let $\chi=\{\ket{\psi_i}\}_{i=1}^{N}\subset\mathbb{C}^{d}$ be a uniform frame and set
$\Pi_i:=\outerproduct{\psi_i}{\psi_i}$. If $\chi$ is a $1$-design, then
\begin{equation}
    \frac{1}{N}\sum_{i=1}^{N}\Pi_i=\frac{I}{d},
\end{equation}
and therefore the operators
\begin{equation}
    E_i:=\frac{d}{N}\Pi_i
\end{equation}
form an equal-trace\footnote{The proof of Theorem~\ref{Th:VS-WB} also applies to weighted $(k-2)$-designs. Accordingly, the arguments below extend to weighted projective $2$-designs by replacing the uniform weights $1/N$ by normalized weights $w_i$ and taking $E_i=dw_i\Pi_i$.} rank-one POVM. General informationally complete POVMs can be used for shadow estimation, but we restrict here to projective $2$-designs, because for this class the canonical linear reconstruction takes the particularly simple form~\cite{Scott2006TightIC,Innocenti2023}
\begin{equation}
    \widehat{\rho}_i=(d+1)\Pi_i-I.
    \label{eq:shadow-reconstruction}
\end{equation}
Both SICs and complete sets of MUBs belong to this class.

When the input state is $\rho$, outcome $i$ occurs with probability
\begin{equation}
    p_i=\Tr(\rho E_i).
\end{equation}
For a Hermitian observable $O$, the random variable
\begin{equation}
    \widehat{o}_i:=\Tr(O\widehat{\rho}_i)
\end{equation}
is an unbiased single-shot estimator of $\Tr(O\rho)$, namely
\begin{equation}
    \mathbb{E}_{i\sim p_\rho}[\widehat{o}_i]
    =
    \Tr(O\rho).
\end{equation}
Since the identity component of $O$ is reconstructed deterministically, it does not contribute to the variance, so we assume without loss of generality that $\Tr(O)=0$. The performance of the estimator is quantified by its variance,
\begin{align}
    \operatorname{Var}_{\chi}(\widehat{o}\mid\rho)
    &:=
    \mathbb{E}_{i\sim p_\rho}
    \left[
        \left(
            \widehat{o}_i-\Tr(O\rho)
        \right)^{2}
    \right]
    \nonumber\\
    &=
    \mathbb{E}_{i\sim p_\rho}
    \left[
        \widehat{o}_i^{\,2}
    \right]
    -
    \Tr(O\rho)^{2}.
    \label{eq:shadow-variance-definition}
\end{align}
The shadow norm associated with this measurement and reconstruction is defined by~\cite{Huang2020Predicting,Innocenti2023}
\begin{equation}
    \|O\|_{\mathrm{sh},\chi}^{2}
    :=
    \sup_{\rho}
    \mathbb{E}_{i\sim p_\rho}
    \left[\widehat{o}_i^{\,2}\right].
    \label{eq:shadow-norm-definition}
\end{equation}
It therefore bounds the variance uniformly over input states. For the average of $L$ independent shots, the variance is reduced by a factor of $L$.

To collect the second moments of the estimators for all input states and observables, define
\begin{equation}
    K_{\chi}
    :=
    \sum_{i=1}^{N}
    E_i\otimes\widehat{\rho}_i\otimes\widehat{\rho}_i.
    \label{eq:shadow-second-moment-operator}
\end{equation}
Then
\begin{equation}
    \mathbb{E}_{i\sim p_\rho}
    \left[
        \widehat{o}_i^{\,2}
    \right]
    =
    \Tr\!\left[
        (\rho\otimes O\otimes O)K_{\chi}
    \right].
    \label{eq:shadow-second-moment}
\end{equation}

For comparison, consider the covariant Haar POVM
\begin{equation}
    d\Pi_{\psi}\,d\psi,
    \qquad
    \Pi_{\psi}:=\outerproduct{\psi}{\psi},
\end{equation}
where $d\psi$ denotes the normalized Haar measure on pure states. The corresponding reconstruction operator is
\begin{equation}
    \widehat{\rho}_{\psi}
    :=
    (d+1)\Pi_{\psi}-I,
\end{equation}
and we define
\begin{equation}
    K_{\mathrm H}
    :=
    d\int_{\mathbb{P}\mathbb{C}^{d}}
    \Pi_{\psi}\otimes
    \widehat{\rho}_{\psi}\otimes
    \widehat{\rho}_{\psi}
    \,d\psi.
    \label{eq:Haar-shadow-second-moment-operator}
\end{equation}

If $\chi$ is an exact projective $3$-design, then its first three moments agree with their Haar counterparts, and hence
\begin{equation}
    K_{\chi}=K_{\mathrm H}.
\end{equation}
The explicit expression for the Haar second moment derived in Refs.~\cite{Huang2020Predicting,Innocenti2023} implies
\begin{align}
    \sup_{\rho}
    \operatorname{Var}_{\mathrm H}(\widehat{o}\mid\rho)
    &\leq
    \|O\|_{\mathrm{sh},\mathrm H}^{2}
    \nonumber\\
    &\leq
    \Tr(O^{2})+2\|O^{2}\|_{\infty}
    \nonumber\\
    &\leq
    3\|O\|_{2}^{2}.
    \label{eq:three-design-shadow-bound}
\end{align}
Therefore, for observables normalized by $\|O\|_{2}=1$, the single-shot variance is bounded by a universal constant, independently of the Hilbert-space dimension.

This observation is particularly relevant for multipartite systems. For a system of $n$ qudits of prime local dimension $q$, the total dimension is $d=q^{n}$. When $q=2$, the set of stabilizer states forms an exact projective $3$-design and can be sampled efficiently using Clifford circuits. By contrast, for odd prime $q>2$, qudit stabilizer states form a projective $2$-design but not a projective $3$-design~\cite{KuengGross2015Stabilizer}. Thus, the standard efficient stabilizer construction does not directly provide an exact $3$-design in this setting.

It is therefore natural to ask whether a projective $2$-design that only approximates a $3$-design can nevertheless retain a dimension-independent variance bound. Let $\chi$ be an arbitrary projective $2$-design. Since its first and second frame operators coincide with the corresponding Haar moments, expanding Eqs.~\eqref{eq:shadow-second-moment-operator} and~\eqref{eq:Haar-shadow-second-moment-operator} shows that the only difference between $K_{\chi}$ and $K_{\mathrm H}$ arises from the third moment~\cite{Innocenti2023}. More precisely,
\begin{equation}
    K_{\chi}-K_{\mathrm H}
    =
    d(d+1)^{2}
    \left(
        \mathcal{F}_{3}(\chi)-\rho_{3}
    \right).
    \label{eq:shadow-K-third-moment}
\end{equation}

Since both the finite and Haar estimators are unbiased, the difference between their variances equals the difference between their second moments. Equations~\eqref{eq:shadow-second-moment} and~\eqref{eq:shadow-K-third-moment} give
\begin{align}
    &\left|
        \operatorname{Var}_{\chi}(\widehat{o}\mid\rho)
        -
        \operatorname{Var}_{\mathrm H}(\widehat{o}\mid\rho)
    \right|
    \nonumber\\
    &=
    d(d+1)^{2}
    \left|
        \Tr\!\left[
            (\rho\otimes O^{\otimes 2})
            \left(
                \mathcal{F}_{3}(\chi)-\rho_{3}
            \right)
        \right]
    \right|
    \nonumber\\
    &\leq
    d(d+1)^{2}
    \|\mathcal{F}_{3}(\chi)-\rho_{3}\|_{2}
    \|O\|_{2}^{2},
    \label{eq:approximate-shadow-variance-difference}
\end{align}
where we used the Hilbert--Schmidt Cauchy--Schwarz inequality and $\|\rho\|_{2}\leq 1$. The same estimate applies to the difference between second moments. Taking the supremum over $\rho$ and using Eq.~\eqref{eq:three-design-shadow-bound} therefore yields
\begin{align}
    \sup_{\rho}
    \operatorname{Var}_{\chi}(\widehat{o}\mid\rho)
    &\leq
    \|O\|_{\mathrm{sh},\chi}^{2}
    \nonumber\\
    &\leq
    \left[
        3+
        d(d+1)^{2}
        \|\mathcal{F}_{3}(\chi)-\rho_{3}\|_{2}
    \right]
    \|O\|_{2}^{2}.
    \label{eq:approximate-shadow-variance-bound}
\end{align}

Consequently, for any dimension-independent constant $\delta>0$, the condition
\begin{equation}
    \|\mathcal{F}_{3}(\chi)-\rho_{3}\|_{2}
    \leq
    \frac{\delta}{d(d+1)^{2}}
    \label{eq:sufficient-shadow-condition}
\end{equation}
is sufficient to guarantee
\begin{equation}
    \sup_{\rho}
    \operatorname{Var}_{\chi}(\widehat{o}\mid\rho)
    \leq
    \|O\|_{\mathrm{sh},\chi}^{2}
    \leq
    (3+\delta)\|O\|_{2}^{2}.
\end{equation}

Although not necessary for a small shadow norm, Eq.~\eqref{eq:sufficient-shadow-condition} provides a simple certificate valid for all input states and observables, expressed directly in terms of the approximation measure studied in this work.

Corollary~\ref{Th:VS-WB3} implies that this sufficient condition can hold only if the cardinality is large enough. An elementary calculation gives
\begin{align}
    N:=\abs{\chi}
    &\geq
    \frac{d^{2}(d+1)}{2}
    -
    \frac{\delta^{2}(d+2)^{2}}{4(d+1)^{2}}
    \nonumber\\
    &=
    \frac{d^{3}}{2}+O(d^{2}),
\end{align}
where the asymptotic expression holds for fixed $\delta$ as $d\to\infty$. We note that a recent paper \cite{MinimalMeasurementSettings2026}, proved that the minimum number of measurement bases needed for worst-case optimal shadow estimation scales as $O(d^{2})$, consistently with our condition of at least $ \frac{d^{3}}{2}+O(d^{2})$ outcomes. 
}

{
\hypersetup{linkcolor=red,citecolor=red,urlcolor=red}
\section{Discussion}\label{sec:conclusion}

Our work focused on the question of how well finite ensembles of quantum states can approximate Haar randomness. We have developed a formal and quantitative method for addressing this question through the Welch bounds. To do so, a first necessary step was to derive stronger Welch bounds, which are informative in the regime where the number of states in the ensemble is very limited. This provides a quantitative and operationally meaningful notion of how well a finite ensemble of states approximates Haar randomness. 

A case of particular interest is when the ensemble of states forms a measurement or a set of measurements, i.e., a $1$-design. Here, our approach helps to quantify the quality of the set as a higher-order design. Specifically, complete sets of MUBs and SICs saturate our new Welch bounds, implying that they are optimal approximate $3$-designs among exact $1$-designs of the same cardinality, in the average-error metric. This case also turns out to be particularly relevant in the context of shadow tomography. Here, our approximation error connects directly to the fluctuations of the standard single-shot estimators, as captured by the notion of shadow norm. This gives a sufficient condition on the third-moment error for dimension-independent estimator variance, together with a necessary lower bound on the number of measurement outcomes for that condition to hold. It would be interesting to explore this connection further and extend it to more general problems in state estimation and benchmarking~\cite{Zhu2022,GILLES2026}.

A key technical ingredient in our work is the explicit characterisation of the full spectrum (including
multiplicities) of the partially transposed Haar moment operator. This may find further applications beyond the context of the present work. 

More generally, the methods behind our improved Welch bounds can also be applied to other types of designs to quantify how well finite ensembles can approximate Haar-random ones. We develop the examples of toric designs and unitary designs, which are discussed in detail in Appendices~\ref{app:toric-designs} and~\ref{app:unitary-designs}, respectively. For unitary designs, the idea is to represent each unitary by its Choi state and to characterize the corresponding ensemble. An interesting question for the future is to fully characterize the spectral information of the corresponding moment operators, similarly to our explicit construction for the case of state designs. It would also be interesting to determine to which other types of designs our methods apply, particularly when additional physical constraints must be taken into account. Of particular interest are Scrooge designs \cite{mok2026}, which have recently been introduced in the context of many-body systems and deep thermalisation \cite{Mark2024}. Another direction concerns states generated by circuits of limited complexity \cite{Riddell} or states with limited entanglement \cite{Yu2021a,Yu2021}.

Finally, it would be interesting to see if our new Welch bounds find applications in other areas, notably in signal processing, where the Welch bound was originally derived.

}

\begin{acknowledgements}
We thank the organizers and speakers of the \lq\lq Representation Theory in Quantum Information Science" Masterclass at Copenhagen University during which the original idea for this work emerged. We acknowledge support from the Swiss National Science Foundation (NCCR SwissMAP). DG acknowledges support by NWO grant NGF.1623.23.025 (“Qudits in theory and experiment”).
\end{acknowledgements}

\bibliography{biblio.bib}

\newpage
\appendix
\onecolumngrid

\section{List of symbols} \label{app:symbols}

\begin{description}
  \item[$k$] Number of systems
  \item[$d$] Dimensionality of the system
  \item[$X^{\mathsf{\Gamma}_{i_{1},i_{2},\dots}}$] Partial transpose of an operator $X$ over systems $i_{1},i_{2},\dots$
  \item[$X^{\mathsf{\Gamma}}$] Partially transposed operator
  \item[$\mathcal{F}_{n}$] $n$-fold tensor product frame operator
  \item[$\mathcal{F}_{n,m}$] $(n,m)$-fold tensor product frame operator, i.e.\ $\mathcal{F}^{\mathsf{\Gamma}}_{n+m}$
  \item [$\vee^{k}\C^d$] Symmetric subspace of $(\mathbb{C}^{d})^{\otimes k}$. 
  \item[$\Pi_n$] Projector onto the symmetric subspace of $n$ qudits, i.e., $\vee^{k}\C^d$
  \item[$\rho_n$] Maximally mixed state on the symmetric subspace of $n$ qudits
  \item[$\rho_{n,m}$] Partial transpose over the last $m$ systems of $\rho_{n+m}$
  \item[$V_r$] Invariant subspace of $\mathcal{U}_{n,m}^{d}$
  \item[$\Pi^{(r)}$] Projector on the $V_{r}$ subspace.
  \item[$\chi$] Frame, i.e.\ a set of unit vectors in $\mathbb{C}^d$
  \item[$\mathcal{E}_k(\chi)$] $k$-energy of a frame, equal to
  $|\chi|^{-2}\sum_{i=1}^{|\chi|}|\langle\psi_i\mid\psi_j\rangle|^{2k}$
\end{description}

\section{Proof of Eq.~\eqref{Eq:GUEIdentity}} \label{Sec:Eq(18)proof}
Matrices from the normalized Ginibre ensemble can equivalently be viewed as Haar-uniform unit vectors once the complex vector space of matrices supported on the symmetric subspace, i.e.,
$\mathcal{L}[\vee^{k}\C^{d}]$, is identified with $\C^{D^2}$ where {$D:=\dim[\vee^{k}\C^{d}]$} (via vectorization and the Hilbert--Schmidt inner product). Using the suggestive notation $\ket{M},\ket{A}$ for $M,A\in \mathcal{L}[\vee^{k}\C^{d}]$, we obtain (for an Hermitian $A=A^\dagger$):
\begin{align}
    & \underset{M\sim \mu}{\mathbb{E}}
  \left[
    \abs{\Tr\!\left[\frac{M}{\lVert M\rVert_{2}}A\right]}^{2}
  \right]= \underset{\ket{M}\sim Haar}{\mathbb{E}}
  \left[
    \abs{\braket{A}{M}}^{2}
  \right]= \underset{\ket{M}\sim Haar}{\mathbb{E}}
  \left[\braket{A}{M}\braket{M}{A}
  \right]={D^{-2}}\langle A| {I_{D^2}}|A\rangle := \frac{\Tr[AA^{\dagger}]}{{D^2}} \,.
\end{align}
Here, $I_{{D^2}}$ is the ${D^2\cross D^2}$ identity and we used once again the identity $\underset{\ket{M}\sim Haar}{\mathbb{E}}
  \left[ \ketbra{M} \right]=\frac{I_{{D^2}}}{{D^2}}$.
The proof is concluded by noticing $\Tr[AA^{\dagger}]=\lVert A \rVert_{2}^{2}$.
These calculations also show that the specific choice of the ensemble is not crucial. In particular, the proof holds as long as the probability measure $\mu'$ over $\mathcal{L}[\vee^{k}\C^{d}]$ has the property
\[
\underset{M\sim \mu'}{\mathbb{E}}
\left[ \frac{\ketbra{M}}{\lVert M\rVert_{2}^{{2}}} \right]=\frac{I_{{D^2}}}{{D^2}}.
\]

\section{Spectrum of the Partially Tranposed symmetric subspace projector} \label{app:spectrum}

The symmetric subspace projector $\Pi_{n}$ is an important operator, which found many applications in quantum information theory \cite{harrow2013churchsymmetricsubspace}. 
It admits several equivalent definitions:
\begin{align}
    \Pi_{n} &:= \frac{1}{n!} \sum_{\pi \in \S_n} R_\pi = \binom{n+d-1}{n} \rho_n \\
    \rho_n &:= \int_{\mathrm{Haar}} \ketbra{\psi}\xp{n} \, d\psi.
\end{align}
where $R_\pi \in \End((\C^d)\xp{n})$ is the tensor representation of the symmetric group element $\pi \in \S_n$. There is a convenient basis of the symmetric subspace, which consists of symmetrized states which are labelled by weights $w$ defined as
\begin{equation}
    w := (w_1,\dotsc,w_d), \quad \sum_{i=1}^d w_i = n, \quad w_i \geq 0.
\end{equation}
Denote the set of weights by $\mathsf{W}_{n,d}$. Its cardinality is $\abs{\mathsf{W}_{n,d}} = \binom{n+d-1}{n}$. The basis then is defined for every $w \in W_{n,d}$ as
\begin{equation}
   \ket{w} := \sqrt{\frac{\prod_{i=1}^d w_i!}{n!}} \sum_{\substack{x \in [d]^n \\ \mathsf{wt}(x)=w}} \ket{x}  ,
\end{equation}
where $\mathsf{wt}(x)$ is a weight of the string $x$. So one can also write
\begin{equation}
    \Pi_{n} = \sum_{w \in W_{n,d}} \ketbra{w}.
\end{equation}

Consider now $n+m$ qudit Hilbert space. We would like to find spectrum of the operator, defined by partially transposing the operator $\rho_{n+m}$ on the last $m$ systems:
\begin{align}
    \rho_{n,m} := (\rho_{n+m})^{\mathsf{\Gamma}_{n+1,\dotsc,n+m}}.
\end{align}
Moreover, using integral definition of $\rho_n$ we can write
\begin{align} \label{Eq:ProjectorHaarIdentity}
    \rho_{n,m} = \int_{\mathrm{Haar}} \ketbra{\psi}\xp{n} \otimes  \ketbra{\bar{\psi}}\xp{m} \, d\psi. 
    % =  \int_{\mathrm{Haar}} (U\proj{0}U^\dagger)^{\otimes n} \otimes (\bar{U}\proj{0}\bar{U}^\dagger)^{\otimes m} \, dU.
\end{align}
Note that $\rho_{k,0} = \rho_{0,k} = \rho_{k}$.

\subsection{Representation theory preliminaries}

To obtain the spectrum of $\rho_{n,m}$, our main technical ingredient is Schur--Weyl duality, and its generalisation---mixed Schur--Weyl duality. In this section, we provide a succinct summary of the necessary facts, for more details see \cite{grinko2023gt,grinko2025mixed}.

To see the why we need (mixed) Schur--Weyl duality, it is instructive to view the operator $\rho_{n,m}$ as the sum of all partially transposed elements of the symmetric group $\S_{n+m}$:
\begin{equation}
    \rho_{n,m} = \frac{1}{\binom{n+m+d-1}{n+m}} \frac{1}{(n+m)!}\sum_{\pi \in \S_{n+m}} R_\pi^{\mathsf{\Gamma}},
\end{equation}
where $R_\pi \in \End((\C^d)\xp{n})$ is the natural tensor representation of the symmetric group $\S_{n+m}$: 
\begin{equation}
    R_\pi := \sum_{x \in [d]^{n+m}} \ketbra{x_{\pi^{-1}(1)},..x_{\pi^{-1}(n+m)}}{{x_{1},..x_{n+m}}}.
\end{equation}
Elements $R_\pi^{\mathsf{\Gamma}}$ span a so-called matrix algebra of partially transposed permutations $\mA_{n,m}^d$: 
\begin{equation}
    \mA_{n,m}^d := \text{span}_{\mathbb{C}}\{ R_{\pi}^{\mathsf{\Gamma}}: \pi \in S_{n+m} \}.
\end{equation}
This algebra was studied extensively over the past years, and its representation theory is known, see \cite{grinko2023gt,grinko2025mixed,MozrzymasHorodeckiStudzinski2014Structure,StudzinskiHorodeckiMozrzymas2013Commutant,MozrzymasStudzinskiHorodecki2018SimplifiedFormalism,HorodeckiStudzinskiMozrzymas2025Iterative,MozrzymasHorodeckiStudzinski2024Frobenius,StudzinskiMylnikMozrzymasHorodeckiGrinko2025}.
The operator $\rho_{n,m}$ is a ``natural'' element within this algebra, so it we should expect its spectrum to be understood from the representation theory of the algebra $\mA_{n,m}^d$. 

To that end, note that $\rho_{n,m}$ commute with all unitaries of the form $U\xp{n}\otimes\bar{U}\xp{m}$ for every $U \in \U(d)$. This naturally brings us to the concept of mixed Schur--Weyl duality. Define the following matrix algebra:
\begin{equation}
    \mU_{n,m}^d := \text{span}_{\mathbb{C}}\{ U\xp{n}\otimes\bar{U}\xp{m}: U \in \U(d) \}.
\end{equation}
Mixed Schur--Weyl duality is a statement that algebras $\mU_{n,m}^d$ and $\mA_{n,m}^d$ are mutual commutants:
\begin{align}
    \mA_{n,m}^d &= \End_{\mU_{n,m}^d}((\C^d)\xp{n+m}), \\
    \mU_{n,m}^d &= \End_{\mA_{n,m}^d}((\C^d)\xp{n+m}).
\end{align}
To restate this in a different language, let's recall a special case ($m=0$), known as Schur--Weyl duality. It states that these algebras can be seen as two natural representations of groups $\S_n$ and $\U(d)$, acting on the space $(\C^d)\xp{n}$. Then the space $(\C^d)\xp{n}$ decomposes into direct sum of irreducible representations of these two groups:
\begin{equation}
    (\C^d)\xp{n} \cong \bigoplus_{\lambda \vdash_d n} \mathcal{W}_\lambda \otimes \mathcal{H}_\lambda,
\end{equation}
where $\lambda$ are Young diagrams with $n$ boxes with at most $d$ rows, $\mathcal{W}_\lambda$ and $\mathcal{H}_\lambda$ are irreps of unitary group $\U_d$ and symmetric group $\S_n$ respectively (also called Weyl module and Specht module in the literature). We denote their dimensions by $m_\lambda$ and $d_\lambda$ respectively. In particular, the symmetric subspace corresponds to the horizontal Young diagram $\lambda = (n)$:
\begin{equation}
    \mathcal{W}_{(n)} = \vee^n\C^d,
\end{equation}
and $\mathcal{H}_{(n)}$ is the trivial one-dimensional irrep of the symmetric group.

Similarly, more general mixed Schur--Weyl duality can be formulated as decomposition of the space $(\mathbb{C}^d)^{\otimes n} \otimes (\mathbb{C}^d)^{* \otimes m}$ under the action\footnote{Space $(\mathbb{C}^d)^{*}$ carries the dual action of the unitary group $U(d)$.} of algebras $\mA_{n,m}^d$ and $\mU_{n,m}^d$:
\begin{equation}
    (\mathbb{C}^d)^{\otimes n} \otimes (\mathbb{C}^d)^{* \otimes m} \cong \bigoplus_{\lambda \in \hat{\mA}_{n,m}^d} \mathcal{W}_\lambda \otimes \mathcal{V}_\lambda,
\end{equation}
where here $\mathcal{V}_\lambda$ is now an irrep of partially transposed permutation algebra $\mA_{n,m}^d$ with the set of irrep labels defined as
\begin{equation}
    \hat{\mA}_{n,m}^d := \set{ \lambda \in \mathbb{Z}^d  \, : \, \lambda_1 \geq \dotsc \geq \lambda_d, \, \sum_{i=1}^d \abs{\lambda_i} = n+m-2r, \, r \in \set{0,\dotsc,\min(n,m)}},
\end{equation}
which can be interpreted as staircases, or pairs of Young diagrams. In particular, the isotypic components corresponding to $\lambda = (n-r,0,\dotsc,0,-m+r)$ for every $r \in \{0,1,\dotsc,\min(n,m) \}$ are of particular interest to us. Specifically, consider the restriction of irreps $\mathcal{H}_\lambda$ of $\mA_{n,m}^d$ to the subalgebra $\C[\S_n \times \S_m]$:
\begin{equation}
    \mathcal{V}_\lambda \downarrow^{\mA_{n,m}^d}_{\C[\S_n \times \S_m]} \simeq \bigoplus_{\substack{\mu \vdash_d n \\ \nu \vdash_d m}} \mathcal{H}_\mu \otimes \mathcal{H}_\nu \otimes \C^{c^{\lambda}_{\mu,\nu}},
\end{equation}
where $c^{\lambda}_{\mu,\nu}$ are multiplicities of the restriction (they can be expressed via so-called Littlewood--Richardson coefficients).
In particular, note that according to the Pieri rule for $\lambda = (n-r,0,\dotsc,0,-m+r)$, $\mu = (n)$ and $\nu = (m)$ we have $c^{\lambda}_{\mu,\nu} = 1$. Therefore, we have the following decomposition of the full space according to the joint action of $\C[\S_n \times \S_m]$ and $\mU_{n,m}^d$:
\begin{equation}
    (\mathbb{C}^d)^{\otimes n} \otimes (\mathbb{C}^d)^{* \otimes m} \cong \of[\bigg]{\bigoplus_{r=0}^{\min(n,m)} \mathcal{W}_{(n-r,0,\dotsc,0,-m+r)} \otimes \mathcal{H}_{(n)} \otimes \mathcal{H}_{(m)}} \oplus \dotsb
\end{equation}
where $\mathcal{H}_{(n)}$ and $\mathcal{H}_{(m)}$ are trivial irreps for $\S_n$ and $\S_m$ respectively.
Let's denote the $r$-subspaces under direct summand by $V_r$:
\begin{equation}
\label{eq:V_r_def}
    V_r := \mathcal{W}_{(n-r,0,\dotsc,0,-m+r)} \otimes \mathcal{H}_{(n)} \otimes \mathcal{H}_{(m)},
\end{equation}
and the corresponding orthogonal projector by $\Pi^{(r)} \in \End((\mathbb{C}^d)^{\otimes n+m})$.

%%%%%%%%%%%%%%%%%%%%%%%%%%%%%%%%%%%%%%%%%%%%%%%%
\subsection{Main technical result}

\begin{thm}[Full spectrum of $\rho_{n,m}$] Let $d \geq 2$ and $n,m \geq 0$. Then there are $1+\min(n,m)$ different non-zero eigenvalues of the operator $\rho_{n,m}$, which can be writte written as
\begin{equation}
    \rho_{n,m} = \sumab{r=0}{\min(n,m)} C_{r}(n,m,d)\Pi^{(r)}, 
\end{equation}
where for every $r \in \set{0,1,\dots,\min(n,m)}$ the eigenvalue
\begin{align}
    C_r(n,m,d) &= \binom{n+m+d-1}{r} {\binom{n+m}{m}}^{-1} {\binom{n+m+d-1}{n+m}}^{-1} = \frac{n!m!(d-1)!}{r!(n+m-r+d-1)!}
\end{align}
appears with multiplicity $\eta_r(n,m,d) := \mathrm{dim} \Pi^{(r)}$, given by
\begin{equation}
    \eta_r(n,m,d) = m_{(n-r,0^{d-2},-m+r)} = \binom{n-r+d-2}{d-2} \binom{m-r+d-2}{d-2} \frac{n+m-2r+d-1}{d-1},
\end{equation}
where $m_{(n-r,0^{d-2},-m+r)}$ is the dimension of the highest-weight $(n-r,0^{d-2},-m+r)$ irrep of $\U(d)$.
\end{thm}
\begin{proof}
    Note the following symmetries of the operator $\rho_{n,m}$:
    \begin{align}
        \rho_{n,m} (U\xp{n} \otimes \bar{U}\xp{m}) &= (U\xp{n} \otimes \bar{U}\xp{m}) \rho_{n,m} \quad \forall \; U \in \U(d) \label{eq_sym1}\\
        \rho_{n,m} = \rho_{n,m} (R_\pi \otimes R_\sigma) &= (R_\pi \otimes R_\sigma) \rho_{n,m} \quad \forall \; (\pi,\sigma) \in \S_n \times \S_m \label{eq_sym2}
    \end{align}
    In particular, the second symmetry implies 
    \begin{equation}
        \rho_{n,m} = \rho_{n,m} (\Pi_{n} \otimes \Pi_{m}) = ( \Pi_{n} \otimes \Pi_{m}) \rho_{n,m}
        \label{eq_sym3}
    \end{equation}
    and, in particular, that the operator $\rho_{n,m}$ is supported only on the symmetric subspace product $\mathrm{supp}( \Pi_{n} \otimes \Pi_{m})$.
    Moreover, in the mixed Schur basis it is supported only on highest-weight irreps $\lambda^{(r)} := (n-r,0,\dotsc,0,-m+r)$ for every $r \in \set{0,\dotsc,\min(n,m)}$, which correspond to subspaces $V_r$, see \cref{eq:V_r_def}.

    We now describe (up to normalization) explicit highest-weight vectors due to \cite{benkart1994tensor}:
    \begin{align}
        \ket{\Phi_r} &:= (\Pi_n \otimes \Pi_m) \ket{\omega_r}, \\\ket{\omega_r} &:= \ket{1}\xp{n-r} \otimes \ket{\omega}\xp{r} \otimes \ket{d}\xp{m-r},
    \end{align}
    where $\ket{\omega} := \sum_{i=1}^d \ket{i} \otimes \ket{i}$ is unnormalized EPR state (or, equivalently, vectorisation of the identity operator). It follows from \cite{benkart1994tensor} that the vectors $\ket{\Phi_r}$ generate the subspaces $V_r$ from the action of $\mU_{n,m}^d$. Schur Lemma implies that they are eigenvectors of $\rho_{n,m}$ due to symmetries \cref{eq_sym1,eq_sym2,eq_sym3}. Therefore, to compute the eigenvalues we need to compute Rayleigh quotients:
    \begin{equation}
        C_r(n,m,d) = \frac{\bra{\Phi_r}\rho_{n,m}\ket{\Phi_r}}{\brakett{\Phi_r}{\Phi_r}}.
    \end{equation}
    We compute numerator and the denominator separately. Firstly, we compute the numerator
    \begin{align}
        \bra{\Phi_r}\rho_{n,m}\ket{\Phi_r} &= \bra{\omega_r}(\Pi_{n} \otimes \Pi_{m})\rho_{n,m}(\Pi_{n} \otimes \Pi_{m})\ket{\omega_r} 
        = \bra{\omega_r} \rho_{n,m} \ket{\omega_r} \\
        &= \int_{\mathrm{Haar}} \abs{\brakett{\omega_r}{\psi\xp{n}\otimes\bar{\psi}\xp{m}}}^2 \, d\psi \\
        &= \int_{\mathrm{Haar}} \abs{\psi_1}^{2(n-r)} \abs{\psi_d}^{2(m-r)} d\psi \\
        &= \int_{\mathrm{Dirichlet}} x_1^{n-r} x_d^{m-r} dx = \frac{(n-r)!(m-r)!(d-1)!}{(n+m-2r+d-1)!}
    \end{align}
    where in the last line we used the fact that $x_i := \abs{\psi_i}^{2}$ for $i \in [d]$ are distributed according to the Dirichlet distribution $\mathrm{Dir}(1,\dotsc,1)$.

    Secondly, we compute the norm of $\ket{\Phi_r}$:
    \begin{align}
        \brakett{\Phi_r}{\Phi_r} &= \bra{\omega_r} \Pi_{n} \otimes \Pi_{m} \ket{\omega_r} \\
        &= \sum_{w \in W_{n,d}} \sum_{w' \in W_{m,d}} \abs{\brakett{\omega_r}{w, w'}}^2 \\
        &= \sum_{w \in W_{n,d}} \sum_{w' \in W_{m,d}} \of[\bigg]{\sum_{x \in [d]^r} \brakett{L_x}{w} \brakett{R_x}{w'}}^2\\
        &= \sum_{w \in W_{n,d}} \sum_{w' \in W_{m,d}} \sum_{x \in [d]^r} \sum_{y \in [d]^r} \brakett{L_x}{w} \brakett{L_y}{w} \brakett{R_x}{w'} \brakett{R_y}{w'}
    \end{align}
    where we used
    \begin{align}
        \ket{\omega_r} &= \sum_{x \in [d]^r} \ket{L_x} \otimes \ket{R_x}, \qquad \ket{L_x} := \ket{1^{n-r},x_1,\dotsc,x_r}, \quad \ket{R_x} := \ket{x_r,\dotsc,x_1,d^{m-r}}.
    \end{align}
    However, it is clear that
    \begin{align}
        \sum_{w \in W_{n,d}} \brakett{L_x}{w} \brakett{L_y}{w} &= \delta_{\mathsf{wt}(x),\mathsf{wt}(y)} \brakett{L_x}{\mathsf{wt}(L_x)}^2 = \delta_{\mathsf{wt}(x),\mathsf{wt}(y)} \frac{\prod_{i=1}^d \mathsf{wt}(L_x)_i!}{n!}, \\
        \sum_{w' \in W_{n,d}} \brakett{R_x}{w'}\brakett{R_y}{w'} &= \delta_{\mathsf{wt}(x),\mathsf{wt}(y)} \brakett{R_x}{\mathsf{wt}(R_x)}^2 = \delta_{\mathsf{wt}(x),\mathsf{wt}(y)} \frac{\prod_{i=1}^d \mathsf{wt}(R_x)_i!}{m!},
    \end{align}
    so using the relations between weights of the strings $x$ and $L_x, \, R_x$
    \begin{align}
        \mathsf{wt}(L_x) &= (n-r + \mathsf{wt}(x)_1,\mathsf{wt}(x)_2,\dotsc,\mathsf{wt}(x)_d) \\
        \mathsf{wt}(R_x) &= (\mathsf{wt}(x)_1,\mathsf{wt}(x)_2,\dotsc,m-r+\mathsf{wt}(x)_d),
    \end{align}
    we can write using elementary arithmetic
    \begin{align}
        \brakett{\Phi_r}{\Phi_r} &= \frac{1}{n! m!} \sum_{\substack{x,y \in [d]^r}} \delta_{\mathsf{wt}(x),\mathsf{wt}(y)} \prod_{i=1}^d \prod_{j=1}^d \mathsf{wt}(L_x)_i! \, \mathsf{wt}(R_x)_j! \\
        &= \frac{1}{n! m!} \sum_{x,y \in [d]^r} \delta_{\mathsf{wt}(x),\mathsf{wt}(y)} (n-k+\mathsf{wt}(x)_1)! (m-r+\mathsf{wt}(x)_d)! \prod_{i=2}^d \prod_{j=1}^{d-1} \mathsf{wt}(x)_i! \,  \mathsf{wt}(x)_j! \\
        &= \frac{r!^2}{n! m!} \sum_{x,y \in [d]^r} \delta_{\mathsf{wt}(x),\mathsf{wt}(y)} \frac{(n-r+\mathsf{wt}(x)_1)! \, (m-r+\mathsf{wt}(x)_d)!}{\mathsf{wt}(x)_1! \, \mathsf{wt}(x)_d!} \of*{\frac{\prod_{i=1}^d \mathsf{wt}(x)_i!}{r!}}^2 \\
        &= \frac{r!^2}{n! m!} \sum_{w \in \mathsf{W}_{r,d}} \sum_{\substack{x,y \in [d]^r \\ \mathsf{wt}(x)= w \\ \mathsf{wt}(y) = w}} \frac{(n-r+\mathsf{wt}(x)_1)! \, (m-r+\mathsf{wt}(x)_d)!}{\mathsf{wt}(x)_1! \, \mathsf{wt}(x)_d!} \of*{\frac{\prod_{i=1}^d \mathsf{wt}(x)_i!}{r!}}^2 \\
        &= \frac{r!^2}{n! m!} \sum_{w \in \mathsf{W}_{r,d}} \frac{(n-r+w_1)! \, (m-r+w_d)!}{w_1! \, w_d!} \of*{\frac{\prod_{i=1}^d w_i!}{r!}}^2 \of*{\frac{r!}{\prod_{i=1}^d w_i!}}^2 \\
        &= \frac{r!^2}{n! m!} \sum_{w \in \mathsf{W}_{r,d}} \frac{(n-r+w_1)! \, (m-r+w_d)!}{w_1! \, w_d!} \\
        &= \frac{1}{\binom{n}{r} \binom{m}{k}} \sum_{w \in \mathsf{W}_{r,d}} \binom{n-r+w_1}{n-r} \binom{m-r+w_d}{m-r} = \frac{\binom{n+m-r+d-1}{r}}{\binom{n}{r} \binom{m}{r}},
    \end{align}
    where we simplified the last line as
    \begin{align}
        \sum_{w \in \mathsf{W}_{r,d}} \binom{n-r+w_1}{n-r} \binom{m-r+w_d}{m-r} &= \sum_{\substack{i+j+l = r \\ i,j,l \geq 0}} \binom{n-r+i}{i} \binom{m-r+j}{j} \binom{l+d-3}{l} \\
        &= \binom{n+m-r+d-1}{r},
    \end{align}
    according to the Chu--Vandermonde identity. Combining everything together, we get
    \begin{align}
        C_r(n,m,d) &= \frac{\bra{\Phi_r}\rho_{n,m}\ket{\Phi_r}}{\brakett{\Phi_r}{\Phi_r}} = \frac{(n-r)!(m-r)!(d-1)!}{(d+n+m-2r-1)!} \frac{ \binom{n}{r} \binom{m}{r}}{\binom{n+m-r+d-1}{r}} \\
        &= \binom{n+m+d-1}{r} {\binom{n+m}{m}}^{-1} \binom{n+m+d-1}{n+m}^{-1} \\
        &= \frac{n!m!(d-1)!}{r!(n+m-r+d-1)!}.
    \end{align}
    The multiplicity formula follows trivially from mixed Schur--Weyl duality, and the known dimensions of the unitary highest weight irreps $(n-r,0,\dotsc,0,-m+r)$.
\end{proof}

\subsection{Special case}

Later, when we apply the above technical result, we will only be concerned with the case $n=\lceil \frac{k}{2}\rceil,m=\lfloor \frac{k}{2}\rfloor$, thus for each $k,d$ we fix a specific notation, for $r = 0,1,...m $
 \begin{align}  %\label{Eq:SU(d)Constants}
     & \eta_{r}(k,d):= \eta_{r}(\lceil \tfrac{k}{2}\rceil,\lfloor \tfrac{k}{2}\rfloor,d) = \frac{k-2r + d - 1}{d-1}\,
  \binom{\lceil \frac{k}{2}\rceil-r + d - 2}{\lceil \frac{k}{2}\rceil-r}\,
  \binom{\lfloor \frac{k}{2}\rfloor-r + d - 2}{\lfloor \frac{k}{2}\rfloor-r}, \label{Eq:Multiplicities} \\
     & C_{r}(k,d) := C_{r}(\lceil \tfrac{k}{2}\rceil,\lfloor \tfrac{k}{2}\rfloor,d) = \binom{k+d-1}{r} \binom{k}{\lfloor \frac{k}{2}\rfloor}^{-1}\binom{d+k-1}{k}^{-1}.\label{Eq:Eigenvalues} 
 \end{align}

We further notice an insightful combinatorial identity: 
\begin{equation}\label{Eq:CombinatorialEq}
    D(k,d)-\sumab{r=0}{R}\eta_{r}(k,d)=D(k-2(R+1),d)
\end{equation}
 \begin{proof}
     The thesis is elementary by induction once these identities are noticed: $ D(k,d)-\eta_{0}(k,d)=D(k-2,d)$ and $\eta_{r}(k-2,d)=\eta_{r-1}(k,d)$.
\end{proof}

\section{Proof of \eqref{Eq:RankTraceFrob_inProof}} \label{app:RankTraceLemma}

By von Neumann's trace inequality~\cite{NeummanTraceIneq(1937),VonNTraceIneq(1975)},
for any unitary $U$ the quantity $\Tr[XUYU^\dagger]$ is maximized when $X$ commutes with $UYU^\dagger$.
Since
\[
\|X-UYU^\dagger\|_2^2=\|X\|_2^2+\|Y\|_2^2-2\Tr[XUYU^\dagger],
\]
the Frobenius distance is minimized when $X$ and $Y$ commute. Writing $Y=\sum_{i=1}^m y_i \ketbra{i}$ and
$X=\sum_{i=1}^m x_i \ketbra{i}$ in a common eigenbasis, the constraints become
$y_i\ge0$, $\sum_i y_i=\sum_i x_i$, and at most {$N$} coefficients $x_i$ are non-zero. Hence
\[
\|X-Y\|_2^2=\sum_{i=1}^m (y_i-x_i)^2.
\]
Minimizing under the trace and rank constraints is a standard convex optimization. {On the top $N$ eigenspaces, the Lagrange equation $2(x_i-y_i)+\lambda=0$ gives $x_i=y_i+\Delta_1/N$ for $i\le N$ and $x_i=0$ for $i>N$. Substitution gives}
\[
{
\|X-Y\|_2^2
=\sum_{i=N+1}^m y_i^2+N\left(\frac{\Delta_1}{N}\right)^2
=\Delta_2+\frac{\Delta_1^2}{N}.}
\]
This proves~\eqref{Eq:RankTraceFrob_inProof} and its equality statement.

\section{Detailed proof of Theorem~\ref{Th:VS-WB}} \label{app:Welchproof}

In this Appendix we prove Theorem~\ref{Th:VS-WB}, which we restate below for completeness:

\begin{thm}[Sharpened Welch bound]
Let $\chi$ be a $(k-2)$-design with
$N\ge D(k-2,d)+\frac{\eta_{0}(k,d)}{k+d-1}$. Then
\begin{align}
\|\mathcal F_k(\chi)-\rho_k\|_2^2 &\ge \Delta(k,d,N),\\
\mathcal E_k(\chi) &\ge \binom{d+k-1}{k}^{-1} + \Delta(k,d,N),
\end{align}
where 
\begin{align}
\Delta(k,d,N)
:= C_{0}^2(k,d)\left(s+\frac{s^{2}}{\eta_{0}(k,d)-s}\right),
\\ s:=D(k,d)-N.
\end{align}
Here $C_{0}(k,d)$ and $\eta_{0}(k,d)$ are given in Eqs.\eqref{Eq:Eigenvalues} and \eqref{Eq:Multiplicities} respectively.
\end{thm}

The proof relies on two structural statements about the partially transposed frame operator when $\chi$ is a $k'$-design. We state and prove these first, and then complete the proof of the theorem.

\subsection{Structural properties of $\mathcal F_{n,m}$ for $k'$-designs}

\begin{prop}[Block agreement]\label{Prop:BlockDiagonalFrameOp}
Let $\chi=\{\ket{\psi_i}\}_{i=1}^N\subset\C^d$ be a $k'$-design, with $0\le k'\le k$, and let
$n,m\ge 0$ satisfy $n+m=k$ and $n\ge m$.
For all $r,r'\in\{0,1,\dots,m\}$, all $\ket{\Psi_r}\in V_r$, and all $\ket{\Psi_{r'}}\in V_{r'}$,
\begin{equation}
\bra{\Psi_r}\mathcal F_{n,m}\ket{\Psi_{r'}}
=
\bra{\Psi_r}\rho_{n,m}\ket{\Psi_{r'}}
\qquad\text{whenever } r+r'\ge k-k'.
\end{equation}
Here $\vee^n\C^d\otimes \vee^m\C^d=\bigoplus_{r=0}^m V_r$ is the mixed Schur--Weyl decomposition
(see Appendix~\ref{app:spectrum}).
\end{prop}

In words: for a \(k'\)-design, \(\mathcal F_{n,m}(\chi)\) coincides with \(\rho_{n,m}\) on every block pair \((V_r,V_{r'})\) such that \(r+r'\ge k-k'\); hence only blocks with \(r+r'<k-k'\) can differ from Haar.

As an immediate consequence, writing $\Pi^{(r)}$ for the projector onto $V_r$ and
$M_{r,r'}:=\Pi^{(r)}\mathcal F_{n,m}(\chi)\Pi^{(r')}\in\mathrm{Hom}(V_{r'},V_r)$, we have

\begin{equation}\label{Eq:FrameOpDecomp2}
\mathcal{F}_{n,m}(\chi)
=\suma{r+r'< k-k'}M_{r,r'}(\chi)
+ \ \bigoplus_{r=k-k'}^{m} C_r(n,m,d)\,\Pi^{(r)}.
\end{equation}

Equivalently: for $r+r'\ge k-k'$, off-diagonal blocks vanish ($r\neq r'$), while diagonal blocks are fixed to
$C_r(n,m,d)\Pi^{(r)}$, independently of $\chi$. To make the decomposition concrete, we display the first nontrivial case $(k,k')=(4,2)$ in Figure~\ref{fig:block-k4-kp2} in the main text.

\begin{proof}[Proof of Proposition~\ref{Prop:BlockDiagonalFrameOp}]
Following the discussion in Appendix \ref{app:spectrum}, every vector $ \ket{\Psi_{r}} \in V_r$ can be written as
\begin{equation}\label{Eq:Psi_r_form}
\ket{\Psi_{r}}
=
\Pi_{n}\otimes \Pi_{m}
\bigl(\ket{\phi_{n-r}}\otimes \ket{\omega}^{\otimes r}\otimes \ket{\theta_{m-r}}\bigr)
=: \ket{\phi_{n-r},\omega^{r},\theta_{m-r}},
\end{equation}
where $\Pi_{n}$ and $\Pi_{m}$ are the projectors onto the symmetric subspaces of the first $n$ and last
$m$ systems, respectively, $\ket{\phi_{n-r}},\ket{\theta_{m-r}}$ belong to the corresponding symmetric
subspaces and $\ket{\omega}:=\sumab{i=1}{d}\ket{ii}$ is the (unnormalized) maximally entangled state. Since both $\rho_{n,m}$ and $\mathcal F_{n,m}$ are self-adjoint we may assume $r'\ge r$ without loss of
generality.

Using $\mathrm{supp}(\mathcal F_{n,m})\subseteq \vee^{n}\C^d\otimes \vee^{m}\C^d$ and the identity
$\,\bra{\varphi,\bar\varphi}\omega\rangle=\Tr[\ketbra{\varphi}]$ we obtain
\begin{align}
\bra{\Psi_{r'}}\mathcal{F}_{n,m}\ket{\Psi_{r}}
&=
\frac{1}{N}\sum_{i=1}^{N}
\bra{\phi'_{n-r},\omega^{r'},\theta'_{m-r}}
\ketbra{\psi_{i}}^{\otimes n}\otimes \ketbra{\bar{\psi}_{i}}^{\otimes m}
\ket{\phi_{n-r},\omega^{r},\theta_{m-r}}
\nonumber\\
&=
\,\frac{1}{N}\sum_{i=1}^{N}
\bra{\phi'_{n-r},\theta'_{m-r}}
\Bigl(
\ket{\psi_{i}}^{\otimes n-r'}\bra{\psi_{i}}^{\otimes n-r}
\otimes
\ket{\bar{\psi}_{i}}^{\otimes m-r'}\bra{\bar{\psi}_{i}}^{\otimes m-r}
\Bigr)
\ket{\phi_{n-r},\theta_{m-r}}.
\label{Eq:FramOpStep1}
\end{align}

The operator inside the sum is (anti)-isomorphic \footnote{Mathematically, this step is the vectorization (Choi--Jamiołkowski) identification
$\mathrm{vec}:\mathrm{L}(\mathbb C^d)\to \mathbb C^d\otimes\mathbb C^d$,
$\mathrm{vec}(\ketbra{a}{b})=\ket{a}\otimes\ket{\bar b}$, equivalently
$\mathrm{vec}(X)=(X\otimes I)\ket{\omega}$ with
$\ket{\omega}:=\sum_{j=1}^d\ket{j}\otimes\ket{j}$.
In physicist language, this is the usual “raising/lowering indices” (bra--ket dualization). In quantum-information language, it is the Choi/vectorization map. It preserves Hilbert--Schmidt pairings, so operator identities are equivalent to vector identities; e.g.
$\sum_i\ketbra{\psi_i}=I \iff \sum_i \ket{\psi_i}\otimes\ket{\bar\psi_i}=\ket{\omega}$.} to a moment
operator of total order $k-(r+r')$. Since $\chi$ is a $k'$-design and $r+r'\ge k-k'$, we may replace the discrete sum
by the Haar integral, yielding
\begin{align}
\bra{\Psi_{r'}}\mathcal{F}_{n,m}\ket{\Psi_{r}}
&=
\,
{\bra{\phi'_{n-r'},\theta'_{m-r'}}}
\left(
\int_{\mathbb{P}\C^d} d\psi\;
\ket{\psi}^{\otimes n-r'}\bra{\psi}^{\otimes n-r}
\otimes
\ket{\bar{\psi}}^{\otimes m-r'}\bra{\bar{\psi}}^{\otimes m-r}
\right)
\ket{\phi_{n-r},\theta_{m-r}}.
\label{Eq:FramOpPassage}
\end{align}
On the other hand, using $\mathrm{supp}(\rho_{n,m})\subseteq \vee^{n}\C^d\otimes \vee^{m}\C^d$ and
Eq.~\eqref{Eq:HaarMomentEqualsRho}, we similarly obtain
\begin{align}
\bra{\Psi_{r'}}\rho_{n,m}\ket{\Psi_{r}}
&=
\,
{\bra{\phi'_{n-r'},\theta'_{m-r'}}}
\left(
\int_{\mathbb{P}\C^d} d\psi\;
\ket{\psi}^{\otimes n-r'}\bra{\psi}^{\otimes n-r}
\otimes
\ket{\bar{\psi}}^{\otimes m-r'}\bra{\bar{\psi}}^{\otimes m-r}
\right)
\ket{\phi_{n-r},\theta_{m-r}},
\end{align}
which coincides with~\eqref{Eq:FramOpPassage}.
\end{proof}

To further constrain the remaining degrees of freedom in~\eqref{Eq:FrameOpDecomp2}, we show that the trace weight of each isotypic component is fixed.

\begin{lem}[Isotypic weights are preserved by twirling]\label{lem:IsotipicWeigth}
Let $\mathcal F_{n,m}$ be the partially transposed frame operator of a generic frame $\chi$, and let $\Pi^{(r)}$ denote the
projector onto the isotypic component labelled by $\lambda^{(r)}$ for the action of $\mathcal U_{n,m}^d$. Then
\begin{equation}\label{Eq:IsotypicWeight}
\Tr[\mathcal{F}_{n,m}\Pi^{(r)}]=\Tr[\rho_{n,m}\Pi^{(r)}].
\end{equation}
\end{lem}

\begin{proof}
Since $[\mathcal U_{n,m}^{d}(U),\Pi^{(r)}]=0$ for all $U$, cyclicity of the trace gives
\[
\Tr[\mathcal U_{n,m}^{d}(U)\mathcal F_{n,m}\mathcal U_{n,m}^{d}(U^{-1})\Pi^{(r)}]
=
\Tr[\mathcal F_{n,m}\Pi^{(r)}].
\]
Averaging over $U\sim H$ therefore implies
$\Tr[\mathcal F_{n,m}\Pi^{(r)}]=\Tr[\mathcal F_{n,m}^{S}\Pi^{(r)}]$, where
\[
\mathcal F_{n,m}^{S}
:=
\underset{U\sim H}{\mathbb{E}}
\bigl[\mathcal U_{n,m}^{d}(U)\mathcal F_{n,m}\mathcal U_{n,m}^{d}(U^{-1})\bigr].
\]
Finally, by Eq.~\eqref{Eq:ProjectorHaarIdentity} we have $\mathcal F_{n,m}^{S}=\rho_{n,m}$, which yields~\eqref{Eq:IsotypicWeight}. 
\end{proof}

\subsection{Full proof}

Using the structural properties of the partially transposed frame operator from the previous subsection, we are in a position to prove prove Theorem~\ref{Th:VS-WB}. As in the proof of Theorem~\ref{Th:StrongerWB}, it suffices to bound
$\|\mathcal F_{n,m}-\rho_{n,m}\|_2^2$ for $n=\lceil k/2\rceil$ and $m=\lfloor k/2\rfloor$. More precisely, both
statements of the theorem are equivalent to
\begin{equation}\label{Eq:GoalApp}
\boxed{\|\mathcal F_{n,m}-\rho_{n,m}\|_2^2 \ge \Delta(k,d,N).}
\end{equation}
We therefore proceed to prove~\eqref{Eq:GoalApp}.
Before the rigorous proof, we give a brief intuition of it. 

\textit{Proof idea.} 
We work in the mixed Schur--Weyl decomposition
$\vee^n\C^d\otimes \vee^{m}\C^d=\bigoplus_{r=0}^{m}V_r$ (with $n=\lceil k/2\rceil$, $m=\lfloor k/2\rfloor$), where
$\rho_{n,m}$ is block diagonal and equals $C_r I$ on each $V_r$.
For a $k'$-design, Proposition~\ref{Prop:BlockDiagonalFrameOp} fixes all blocks
$\Pi_r\mathcal F_{n,m}\Pi_{r'}$ with $r+r'\ge k-k'$ to their Haar values, so only blocks with
$r+r'<k-k'$ can differ from $\rho_{n,m}$.
Hence $\|\mathcal F_{n,m}-\rho_{n,m}\|_2^2$ reduces to an optimization over these remaining blocks,
subject to three constraints inherited from frame operators:
$\mathcal F_{n,m}\succeq 0$, $\operatorname{rank}(\mathcal F_{n,m})\le N$, and fixed isotypic trace weights
(from Lemma~\ref{lem:IsotipicWeigth}).
For $k'=k-1$, only one block remains free and the minimum is given by the rank--trace Frobenius lemma.
For $k'=k-2$, the free part is a $2\times2$ block matrix on $V_0\oplus V_1$; using Schur complement and
a rank identity, we show that under the theorem's size condition the optimizer has zero off-diagonal block,
reducing again to the $k'=k-1$ case.

\begin{proof}
{Set $N:=|\chi|$. If $N\ge D(k,d)$, then $\Delta(k,d,N)\le0$ and the claim is immediate; hence assume $N<D(k,d)$.}

If $k'=k-2$, the only free blocks are $M_{0,0}$ and $M_{0,1}$ (with $M_{1,0}=M_{0,1}^\dagger$), while
$M_{1,1}=C_1\Pi^{(1)}$ is fixed by Proposition~\ref{Prop:BlockDiagonalFrameOp}. All other blocks are fixed (or zero) and therefore do not affect the optimization over free variables. Hence it is enough to restrict to
\[
F=\begin{pmatrix}
M_{0,0} & M_{0,1}\\
M_{0,1}^{\dagger} & C_{1}\Pi^{(1)}
\end{pmatrix},
\qquad
{R:=N+\eta_{0}-D(k,d),}
\]
with constraints
\begin{align}
\operatorname{rank}(F) &{\le R+\eta_1}, \\
F & \succeq 0, \\
\operatorname{Tr}[M_{0,0}] &= C_0\,\eta_0.
\end{align}
{Indeed, the fixed blocks on $\bigoplus_{r\ge2}V_r$ have total rank $D(k,d)-\eta_0-\eta_1$, so the global rank bound $\operatorname{rank}(\mathcal F_{n,m})\le N$ gives the first constraint.}

We now use the $2\times2$ block form of $F$ to rewrite the constraints as
\begin{align}
\operatorname{rank}\!\left(M_{0,0}-\frac{1}{C_{1}}M_{0,1}M_{0,1}^{\dagger}\right) &{\le R}, \label{Eq:rankSC}\\
M_{0,0}-\frac{1}{C_{1}}M_{0,1}M_{0,1}^{\dagger} &\ge 0. \label{Eq:posSC}
\end{align}
Indeed, \eqref{Eq:rankSC} follows from the rank additivity identity (Guttman formula) applied to a block matrix with invertible lower-right block on $V_1$:
\[
\operatorname{rank}(F)
=
\operatorname{rank}(C_1\Pi^{(1)})
+
\operatorname{rank}\!\left(M_{0,0}-M_{0,1}(C_1\Pi^{(1)})^{-1}M_{0,1}^\dagger\right)
=
\eta_1+\operatorname{rank}\!\left(M_{0,0}-\frac{1}{C_1}M_{0,1}M_{0,1}^\dagger\right).
\]
Equation~\eqref{Eq:posSC} is exactly the Schur-complement condition for $F\succeq0$.

Since both the constraints and the objective depend on $M_{0,1}$ only through $M_{0,1}M_{0,1}^{\dagger}$, we set
\[
P:=M_{0,1}M_{0,1}^{\dagger}\succeq0,\qquad M:=M_{0,0},
\]
so that $\|M_{0,1}\|_2^2=\Tr(P)$. Define
\[
f(P,M):=2\,\Tr(P)+\|M-C_0\Pi^{(0)}\|_2^2.
\]
Since all other blocks agree with $\rho_{n,m}$, {$\|\mathcal F_{n,m}-\rho_{n,m}\|_2^2=f(P,M)$.}

{Set $Q:=M-C_1^{-1}P$, $p:=\Tr(P)$, and $t:=\Tr(Q)=C_0\eta_0-p/C_1$. By~\eqref{Eq:rankSC}--\eqref{Eq:posSC}, $Q\succeq0$ and $\operatorname{rank}(Q)\le R$. Hence}
\begin{equation}
{f(P,M)\ge \frac{t^2}{R}-2C_0t+C_0^2\eta_0
+2\left(1-\frac{C_0}{C_1}\right)p=:g(p).}
\end{equation}
{Now $g'(p)=2(1-t/(C_1R))$. By Eqs.~\eqref{Eq:CombinatorialEq} and~\eqref{Eq:Eigenvalues}, the cardinality assumption is exactly $R\ge C_0\eta_0/C_1$, so $g'(p)\ge0$. Therefore}
\begin{equation}
{f(P,M)\ge g(0)=C_0^2\left(\frac{\eta_0^2}{R}-\eta_0\right)=\Delta(k,d,N),}
\end{equation}
{which proves~\eqref{Eq:GoalApp}.}
\end{proof}

\section{Minimal-size projective-toric \texorpdfstring{$t$}{t}-designs and lower bounds for undersized designs}\label{app:toric-designs}

In this section we use the tools and ideas introduced in the main text to answer one of the questions about projective-toric $t$-designs left open in
Ref.~\cite{IosueEtAl2024}. In particular, we prove a stronger lower
bound on the size of projective-toric $t$-designs for odd $t$ and derive a theorem for approximate projective-toric $t$-designs analogous to Theorem~\ref{Th:StrongerWB}.
Before stating the theorems, we give a few definitions.

Let
\begin{equation}
    \mathbb{T}^{m}
    \coloneqq
    (\mathbb{R}/2\pi\mathbb{Z})^{m}
\end{equation}
be the $m$-dimensional torus. The projective torus is the quotient
\begin{equation}
    P(\mathbb{T}^{m})
    \coloneqq
    \mathbb{T}^{m}/\mathbb{T},
\end{equation}
where points differing by the addition of a global phase are identified
\begin{equation}
    (\theta_{1},\ldots,\theta_{m})
    \simeq
    (\theta_{1}+\varphi,\ldots,\theta_{m}+\varphi).
\end{equation}
Thus, $P(\mathbb{T}^{m})$ is naturally identified with the set of
normalised flat vectors in $\mathbb{C}^{m}$ modulo global phase:
\begin{equation}
    \ket{\theta}
    \coloneqq
    \frac{1}{\sqrt m}
    \sum_{j=1}^{m}
    e^{i\theta_{j}}\ket{j}.
    \label{eq:flat-toric-vector}
\end{equation}

\begin{definition}[Degree of a projective character]
Let
\begin{equation}
    \mathbb{Z}_{0}^{m}
    \coloneqq
    \left\{
        v=(v_{1},\ldots,v_{m})\in\mathbb{Z}^{m}
        :
        \sum_{j=1}^{m}v_{j}=0
    \right\}.
\end{equation}
For every $v\in\mathbb{Z}_{0}^{m}$, define its degree as
\begin{equation}
    \deg(v)
    \coloneqq
    \frac{1}{2}\sum_{j=1}^{m}|v_{j}|.
\end{equation}
The vector $v$ defines the character
$\chi_{v}([\phi])\coloneqq e^{i v\cdot\phi}$ on
$P(\mathbb{T}^{m})$\footnote{The projective torus is an abelian group under addition, and $\{\chi_v\}_{v\in\mathbb Z_0^m}$ can be seen as the set of its irreducible characters.} This is well defined because
\begin{equation}
    \chi_{v}([\phi+\theta\mathbf{1}])
    =
    \chi_{v}([\phi])
\end{equation}
for every $\theta\in\mathbb{T}$, since $\sum_{j}v_{j}=0$.
\end{definition}
We use the following equivalent formulation of the definition in
Ref.~\cite{IosueEtAl2024}.

\begin{definition}
\label{def:projective-toric-design}
Let
\begin{equation}
    X =\left\{
        \theta^{(1)},\ldots,\theta^{(N)}
    \right\}\subset P(\mathbb{T}^{m})
\end{equation}
be a finite set. We say that $X$ is a
projective-toric $t$-design if
\begin{align}
    \frac{1}{N}\sum_{r=1}^{N} 
    \exp\left(
        i\sum_{j=1}^{m}z_{j}\theta_{j}^{(r)}
    \right)
    &=\int_{P(\mathbb{T}^{m})}
    e^{i z\cdot\phi}\,
    \mathrm{d}\mu(\phi)\nonumber\\
    &=\delta_{z,0}
\end{align}
for every $z\in\mathbb{Z}_{0}^{m}$ with $\deg(z)\leq t$.
\end{definition}

By construction, a projective-toric $t$-design is also a
projective-toric $(t-1)$-design. As for quantum-state designs,
Definition~\ref{def:projective-toric-design} admits a moment-operator
formulation. Let
\begin{equation}
    \mF_{t}(X)
    \coloneqq
    \frac{1}{N}
    \sum_{r=1}^{N}
    \left(
        \ketbra{\theta^{(r)}}{\theta^{(r)}}
    \right)^{\otimes t},
    \label{eq:toric-frame-operator}
\end{equation}
be the $t$th moment operator, and let
\begin{equation}
    \Omega_{t}
    \coloneqq
    \int_{P(\mathbb{T}^{m})}
    \left(
        \ketbra{\theta}
    \right)^{\otimes t}
    \mathrm{d}\mu(\theta).
    \label{eq:toric-Haar-moment}
\end{equation}
be the corresponding Haar moment operator. Then
\begin{equation}
    X \text{ is a projective-toric $t$-design}
    \quad\Longleftrightarrow\quad
    \mF_{t}(X)=\Omega_{t}.
    \label{eq:toric-moment-characterization}
\end{equation}
The analogy extends to frame energies. Define the $t$th energy of $X$
by
\begin{equation}
    \mathcal{E}_{t}(X)
    \coloneqq
    \frac{1}{N^{2}}
    \sum_{r,s=1}^{N}
    \abs{ \braket{ \theta^{(r)}}{\theta^{(s)}}}^{2t}.
    \label{eq:toric-energy}
\end{equation}

The corresponding Welch-type bound is as follows.

\begin{lemma}[Toric Welch inequality]
\label{lem:toric-Welch}
For every finite set
\begin{equation}
    X\subset P(\mathbb{T}^{m}),
\end{equation}
one has
\begin{equation}
    \mathcal{E}_{t}(X)
    \geq
    C_{t}(m),
    \label{eq:toric-Welch-bound}
\end{equation}
where
\begin{equation}
    C_{t}(m)
    \coloneqq
    \operatorname{Tr}\!\left[\Omega_{t}^{2}\right]
    =
    \frac{1}{m^{2t}}
    \sum_{\substack{\mathbf n\in\mathbb{N}_{0}^{m}\\|\mathbf n|=t}}
    \binom{t}{n_{1},\ldots,n_{m}}^{2}.
    \label{eq:toric-Welch-constant}
\end{equation}
Moreover, equality holds in Eq.~\eqref{eq:toric-Welch-bound} if and
only if $X$ is a projective-toric $t$-design.
\end{lemma}

\begin{proof}
The main idea is to prove the same energy-gap identity as for
quantum-state designs:
\begin{align}
    \left\|\mF_{t}(X)-\Omega_{t}\right\|_{2}^{2}
    &=\mathcal{E}_{t}(X)
    -
    C_{t}(m)\geq 0.
\end{align}
Equation~\eqref{eq:toric-moment-characterization} will then identify the
equality case.

We first compute the Haar moment operator $\Omega_t$. For
$\mathbf{n}=(n_1,\ldots,n_m)\in\mathbb{N}_0^m$ satisfying
$|\mathbf{n}|\coloneqq\sum_{j=1}^m n_j=t$, let $\ket{\mathbf n}$ be the
corresponding normalised occupation-number vector in
$\operatorname{Sym}^{t}(\mathbb{C}^{m})$.
Equation~\eqref{eq:flat-toric-vector} gives
\begin{equation}
    \left\langle
        \mathbf n
        \middle|
        \theta^{\otimes t}
    \right\rangle
    =
    \sqrt{
        \frac{1}{m^{t}}
        \binom{t}{n_{1},\ldots,n_{m}}
    }
    e^{i\mathbf n\cdot\theta}.
    \label{eq:occupation-coefficients}
\end{equation}
Consequently,
\begin{equation}
    \Omega_{t}
    =
    \sum_{|\mathbf n|=t}
    w_{\mathbf n}
    \ketbra{\mathbf n},
    \qquad
    w_{\mathbf n}
    \coloneqq
    \frac{1}{m^{t}}
    \binom{t}{n_{1},\ldots,n_{m}}.
    \label{eq:toric-Haar-diagonalization}
\end{equation}

Second, we show that
$\tr[\mF_t(X)\Omega_t]=\tr[\Omega_t^2]$ for every finite $X$, because
every flat vector has the same occupation-number probabilities.
Therefore, Eqs.~\eqref{eq:toric-frame-operator} and
\eqref{eq:occupation-coefficients} imply
\begin{equation}
    \bra{\mathbf n}\mF_{t}(X)\ket{\mathbf n}
    =
    w_{\mathbf n}
    =
    \bra{\mathbf n}\Omega_{t}\ket{\mathbf n}.
\end{equation}
It follows that
\begin{equation}
    \operatorname{Tr}\!\left[\mF_{t}(X)\Omega_{t}\right]
    =
    \sum_{|\mathbf n|=t}w_{\mathbf n}^{2}
    =
    \operatorname{Tr}\!\left[\Omega_{t}^{2}\right].
\end{equation}
Moreover,
\begin{equation}
    \operatorname{Tr}\!\left[\mF_{t}(X)^{2}\right]
    =
    \mathcal{E}_{t}(X).
\end{equation}
Hence
\begin{align}
    \left\|\mF_{t}(X)-\Omega_{t}\right\|_{2}^{2}
    &=
    \operatorname{Tr}\!\left[\mF_{t}(X)^{2}\right]
    -
    \operatorname{Tr}\!\left[\Omega_{t}^{2}\right]
    \nonumber\\
    &=
    \mathcal{E}_{t}(X)
    -
    C_{t}(m)\geq 0.
    \label{eq:toric-energy-gap}
\end{align}
Equality holds precisely when $\mF_t(X)=\Omega_t$, which is equivalent
to $X$ being a projective-toric $t$-design by
Eq.~\eqref{eq:toric-moment-characterization}.
\end{proof}

In complete analogy with the quantum-state
case, Equation~\eqref{eq:toric-energy-gap} also provides a natural
Hilbert--Schmidt error for approximate projective-toric designs:
\begin{equation}
    \epsilon_{2}^{(t)}(X)
    \coloneqq
    \left\|\mF_{t}(X)-\Omega_{t}\right\|_{2}.
    \label{eq:toric-HS-error}
\end{equation}
Up to a dimension-dependent normalisation, this is the average
moment error considered in the main text. A worst-case notion of approximate
projective-toric design based on individual Fourier coefficients
was considered in Ref.~\cite{IosueEtAl2024}.

As seen above, the Haar $t$-moment operator $\Omega_{t}$ plays a central
role in the theory of projective-toric $t$-designs. In light of the main text, much information can be gained from the
spectra of its partial transposes.

We first introduce some notation. Let
\begin{equation}
    \mathcal{V}_{a,b}^{(m)}
    \coloneqq
    \left\{
        \alpha-\beta:
        \alpha,\beta\in\mathbb{N}_0^m,\ 
        |\alpha|=a,\ 
        |\beta|=b
    \right\},
    \label{eq:definition-Vab}
\end{equation}

\begin{prop}[Spectrum of the partially transposed Haar moment]
    \label{prop:PT-toric-Haar-spectrum}
     For $b\in\{0,\dotsc,t\}$, set $a=t-b$ and let $\mathsf{\Gamma}_{b}$ denote the partial transposition over the last $b$ subsystems. The nonzero eigenvalues of $\Omega_t^{\mathsf{\Gamma}_b}$ are
    \begin{equation}
        \lambda_v
        =
        \frac{1}{m^t}
        \sum_{\substack{|\alpha|=a,\ |\beta|=b\\
                        \alpha-\beta=v}}
        \binom{a}{\alpha}
        \binom{b}{\beta},
        \qquad
        v\in\mathcal{V}_{a,b}^{(m)},
        \label{eq:PT-toric-eigenvalues}
    \end{equation}
    Thus,
    \begin{equation}
        \operatorname{rank}\!\left(\Omega_t^{\mathsf{\Gamma}_b}\right)
        =
        \left|\mathcal{V}_{a,b}^{(m)}\right|.
        \label{eq:PT-toric-rank}
    \end{equation}
\end{prop}

\begin{proof}

For every $v\in\mathcal{V}_{a,b}^{(m)}$, define
    \begin{equation}
        \ket{w_v}
        \coloneqq
        \sum_{\substack{|\alpha|=a,\ |\beta|=b\\
                        \alpha-\beta=v}}
        \sqrt{
            \binom{a}{\alpha}
            \binom{b}{\beta}
        }\,
        \ket{\alpha}\otimes\ket{\beta}.
        \label{eq:wv-definition}
    \end{equation}
    We will prove that
    \begin{equation}
        \Omega_t^{\mathsf{\Gamma}_b}
        =
        \frac{1}{m^t}
        \sum_{v\in\mathcal{V}_{a,b}^{(m)}}
        \ketbra{w_v},
        \label{eq:PT-toric-rank-one-decomposition}
    \end{equation}
    and that the vectors $\ket{w_v}$ are mutually orthogonal.

    First, we have
\begin{equation}
    \Omega_{t}^{\mathsf{\Gamma}_{b}}
    =
    \int_{P(\mathbb{T}^{m})}\ketbra{\phi}
    ^{\otimes a}\otimes\ketbra{\overline{\phi}}
    ^{\otimes b}
    \mathrm{d}\mu(\phi).
    \label{eq:toric-Haar-moment2}
\end{equation}

    In the normalized occupation-number bases one has
    \begin{align}
        \ket{\phi}^{\otimes a}
        \otimes
        \ket{\overline{\phi}}^{\otimes b}
        &=
        \frac{1}{m^{t/2}}
        \sum_{\substack{|\alpha|=a\\|\beta|=b}}
        \sqrt{
            \binom{a}{\alpha}
            \binom{b}{\beta}
        }\,
        e^{\mathrm{i}(\alpha-\beta)\cdot\phi}
        \ket{\alpha}\otimes\ket{\beta}
        \nonumber\\
        &=
        \frac{1}{m^{t/2}}
        \sum_{v\in\mathcal{V}_{a,b}^{(m)}}
        e^{\mathrm{i}v\cdot\phi}\ket{w_v}.
        \label{eq:occupation-expansion-PT}
    \end{align}
    Since every $v\in\mathcal{V}_{a,b}^{(m)}$ satisfies
    $\sum_jv_j=a-b$, the difference $v-v'$ belongs to the root lattice
    $\mathbb{Z}_{0}^{m}$.
   
    Orthogonality of the characters of $P(\mathbb{T}^{m})$ therefore gives
    \begin{equation}
        \int_{P(\mathbb{T}^{m})}
        e^{\mathrm{i}(v-v')\cdot\phi}
        \,\mathrm{d}\mu([\phi])
        =
        \delta_{v,v'}.
    \end{equation}
    Substituting \eqref{eq:occupation-expansion-PT} into
    \eqref{eq:toric-Haar-moment2} proves
    \eqref{eq:PT-toric-rank-one-decomposition}.

    Finally, if $v\neq v'$, the vectors $\ket{w_v}$ and
    $\ket{w_{v'}}$ are supported on disjoint sets of pairs
    $(\alpha,\beta)$, since a given pair determines a unique difference
    $\alpha-\beta$. They are therefore orthogonal. Every
    $\ket{w_v}$ is non-zero by the definition of
    $\mathcal{V}_{a,b}^{(m)}$, so each summand contributes exactly one
    nonzero eigenvalue, proving the claim.
\end{proof}

Armed with these results and definitions, we are ready to prove the
following bound.

\begin{theorem}
    Let $X\subset P(\mathbb{T}^{m})$, and let
    $\nu:X\rightarrow [0,1]$ be a probability measure on it. If
    $(X,\nu)$ is a projective-toric $t$-design, then
    \begin{equation}
        \abs{X}\geq D_{\mathrm{tor}}(m,t):=\abs{\mathcal{V}_{\lfloor t/2\rfloor,\lceil t/2\rceil }^{(m)}}
    \end{equation}
\end{theorem}
\begin{remark}
The cardinalities of the sets $\mathcal{V}_{a,b}^{(m)}$ are known
explicitly. For even degree $t=2k$, the set
$\mathcal{V}_{k,k}^{(m)}$ can be identified with a crystal ball in the
root lattice $\mathbb{Z}_{0}^{m}$, as discussed in
Refs.~\cite{ConwaySloane1997Coordination,IosueEtAl2024}, giving
\begin{equation}
    \abs{\mathcal{V}_{k,k}^{(m)}}
    =
    \sum_{\ell=0}^{k}
    \binom{m-1}{\ell}^{2}
    \binom{m-\ell+k-1}{k-\ell}.
    \label{eq:crystal-ball-number}
\end{equation}
This recovers precisely the bound obtained by a different route in
Proposition~2.13 of Ref.~\cite{IosueEtAl2024}. For odd degree
$t=2k+1$, one instead has
\begin{equation}
    \left|\mathcal{V}_{k+1,k}^{(m)}\right|
    =
    \frac{m}{m-1}
    \sum_{\ell=0}^{m-2}
    \binom{m-1}{\ell}
    \binom{m-1}{\ell+1}
    \binom{k+m-1-\ell}{m-1}.
    \label{eq:odd-affine-rank}
\end{equation}
This is strictly larger than the bound proved in that reference for all
$m\geq 2$. The explicit formula is proved in
Lemma~\ref{lem:odd-cardinality}, which appears below to keep the
discussion flowing.
Indeed, $\mathcal V_{k+1,k}^{(m)}=\bigcup_{j=1}^{m}(\mathbf e_j+\mathcal V_{k,k}^{(m)})$, where $\mathbf e_j$ are the standard basis vectors, and for $m\geq2$ this union strictly contains each individual translate.
\end{remark}

\begin{proof}
   
The proof relies on a simple observation. Let $\mathsf{\Gamma}$ denote the
partial transposition over the last $\lfloor t/2\rfloor$ subsystems. Then
$\text{Rank}[\mF_{t}(X)^{\mathsf{\Gamma}}]\leq \abs{X}$, because the latter
is still a sum of $\abs{X}$ rank-one projectors. If $X$ is a
$t$-design, the moment characterization gives the identity
$\Omega_{t}=\mF_{t}(X)$. Applying $\mathsf{\Gamma}$ to both sides and using
the previous observation gives
\begin{equation}
       X \text{ is a $t$-design} \implies \abs{X}\geq \text{Rank}[\Omega_{t}^{\mathsf{\Gamma}]}.
\end{equation}
The conclusion follows from
Proposition~\ref{prop:PT-toric-Haar-spectrum}, which states that
$\text{Rank}[\Omega_{t}^{\mathsf{\Gamma}}]=
\abs{\mathcal{V}_{\lceil t/2\rceil,\lfloor t/2\rfloor}^{(m)}}$.
\end{proof}

Finally, these tools show that any pair $X$ that is too small to
form a $t$-design must incur a nonzero error in approximating the
$t$th moment. More precisely, we have the following result.

\begin{theorem}
     Let $X\subset P(\mathbb{T}^{m})$, and let
     $\nu:X\rightarrow [0,1]$ be a probability measure on it. {Set $N:=\abs{X}$ and $\mathsf{\Gamma}:=\mathsf{\Gamma}_{\lfloor t/2\rfloor}$. If $N<D_{\mathrm{tor}}(m,t)$, then} the following equivalent lower bounds
     hold:
     \begin{align}
         & \lVert \Omega_{t}-\mF_{t}(X)\rVert_{2}\geq \sqrt{\Delta_{2}+{\frac{\Delta_{1}^{2}}{N}}} \\
         &  \mathcal{E}_{t}(X)-C_{t}(m)\geq \Delta_{2}+{\frac{\Delta_{1}^{2}}{N}}
     \end{align}
     where {$\Delta_{\ell}:=\sumab{i=N+1}{D_{\mathrm{tor}}(m,t)}\lambda_{i}^{\ell}$} and
     $\lambda_{1}\geq\dots\geq\lambda_{{D_{\mathrm{tor}}(m,t)}}$ are the ordered
     eigenvalues of $\Omega_{t}^{\mathsf{\Gamma}}$.
\end{theorem}
\begin{proof}
    Set
    \begin{equation}
        a=\left\lceil\frac{t}{2}\right\rceil,
        \qquad
        b=\left\lfloor\frac{t}{2}\right\rfloor,
    \end{equation}
    and let $\mathsf{\Gamma}=\mathsf{\Gamma}_b$. Define
    \begin{equation}
        A\coloneqq\Omega_t^{\mathsf{\Gamma}},
        \qquad
        B\coloneqq\mF_t(X)^{\mathsf{\Gamma}}.
    \end{equation}
    Partial transposition preserves the Hilbert--Schmidt norm, and
    \begin{equation}
        B
        =
        \frac{1}{N}\sum_{r=1}^{N}\ketbra{\xi_r},
        \qquad
        \ket{\xi_r}
        \coloneqq
        \ket{\theta^{(r)}}^{\otimes a}
        \otimes
        \ket{\overline{\theta^{(r)}}}^{\otimes b}.
    \end{equation}
    Hence $B\geq0$, $\tr(B)=1$, and $\operatorname{rank}(B)\leq N$.

    Let $\mu_1\geq\cdots\geq\mu_N\geq0$ be the {eigenvalues
    of $B$}. By the
    Hoffman--Wielandt inequality \cite{HoffmanWielandt1953},
    \begin{align}
        \left\|\Omega_t-\mF_t(X)\right\|_2^2
        &=
        \|A-B\|_2^2
        \nonumber\\
        &\geq
        \sum_{i=1}^{N}(\lambda_i-\mu_i)^2
        +
        \sum_{i=N+1}^{{D_{\mathrm{tor}}(m,t)}}\lambda_i^2.
        \label{eq:undersized-HW-bound}
    \end{align}
    Since both $A$ and $B$ have unit trace,
    \begin{equation}
        \sum_{i=1}^{N}(\mu_i-\lambda_i)
        =
        1-\sum_{i=1}^{N}\lambda_i
        =
        \Delta_1.
    \end{equation}
    Therefore, by the Cauchy--Schwarz inequality,
    \begin{equation}
        \sum_{i=1}^{N}(\lambda_i-\mu_i)^2
        \geq
        \frac{1}{N}
        \left(
            \sum_{i=1}^{N}(\lambda_i-\mu_i)
        \right)^2
        =
        \frac{\Delta_1^2}{N}.
    \end{equation}
    Substitution into \eqref{eq:undersized-HW-bound} gives
    \begin{equation}
        \left\|\Omega_t-\mF_t(X)\right\|_2^2
        \geq
        \Delta_2+\frac{\Delta_1^2}{N}.
    \end{equation}
    Finally, the weighted version of
    Lemma~\ref{lem:toric-Welch} gives
    \begin{equation}
        \mathcal{E}_t(X)-C_t(m)
        =
        \left\|\Omega_t-\mF_t(X)\right\|_2^2,
    \end{equation}
    proving the second inequality as well.
\end{proof}

\begin{lem}
    \label{lem:odd-cardinality}
    For every $m\geq2$ and $k\in\mathbb{N}_0$,
    \begin{equation}
        \left|\mathcal{V}_{k+1,k}^{(m)}\right|
        =
        \frac{m}{m-1}
        \sum_{\ell=0}^{m-2}
        \binom{m-1}{\ell}
        \binom{m-1}{\ell+1}
        \binom{k+m-1-\ell}{m-1}.
        \label{eq:odd-cardinality}
    \end{equation}
\end{lem}

\begin{proof}
    For $v\in\mathbb{Z}^m$, let
    \begin{equation}
        q(v)\coloneqq\sum_{j=1}^{m}v_j^-,
        \qquad
        v_j^-\coloneqq\max\{-v_j,0\}.
    \end{equation}
    Cancelling the common part of $\alpha$ and $\beta$ shows that
    \begin{equation}
        v\in\mathcal{V}_{k+1,k}^{(m)}
        \quad\Longleftrightarrow\quad
        \sum_{j=1}^{m}v_j=1
        \quad\text{and}\quad
        q(v)\leq k.
        \label{eq:odd-V-characterization}
    \end{equation}
    Consequently,
    \begin{align}
        \sum_{k\geq0}
        \left|\mathcal{V}_{k+1,k}^{(m)}\right|z^k
        &=
        \frac{1}{1-z}
        \sum_{\substack{v\in\mathbb{Z}^m\\\sum_jv_j=1}}
        z^{q(v)}
        \nonumber\\
        &=
        \frac{1}{1-z}
        [x^1]
        \left(
            \sum_{n\in\mathbb{Z}}
            x^n z^{\max\{-n,0\}}
        \right)^m.
        \label{eq:odd-generating-function-start}
    \end{align}
    The single-coordinate generating function is
    \begin{equation}
        \sum_{n\in\mathbb{Z}}
        x^n z^{\max\{-n,0\}}
        =
        \frac{1}{1-x}
        +
        \frac{z/x}{1-z/x}
        =
        \frac{1-z}{(1-x)(1-z/x)}.
    \end{equation}
    Substitution into \eqref{eq:odd-generating-function-start} gives
    \begin{align}
        \sum_{k\geq0}
        \left|\mathcal{V}_{k+1,k}^{(m)}\right|z^k
        &=
        (1-z)^{m-1}
        \sum_{r\geq0}
        \binom{m+r}{r+1}
        \binom{m+r-1}{r}z^r
        \nonumber\\
        &=
        m(1-z)^{m-1}
        {}_2F_1(m,m+1;2;z).
        \label{eq:odd-hypergeometric-short}
    \end{align}
    {Here} ${}_2F_1$ denotes Gauss's hypergeometric function
    \cite[\href{https://dlmf.nist.gov/15.2}{Section~15.2}]{NIST:DLMF}.
    Euler's transformation
    \cite[\href{https://dlmf.nist.gov/15.8.E1}{Eq.~(15.8.1)}]{NIST:DLMF} yields
    \begin{align}
        m(1-z)^{m-1}{}_2F_1(m,m+1;2;z)
        &=
        \frac{m}{(1-z)^m}
        {}_2F_1(1-m,2-m;2;z)
        \nonumber\\
        &=
        \frac{m}{m-1}
        \frac{
            \displaystyle
            \sum_{\ell=0}^{m-2}
            \binom{m-1}{\ell}
            \binom{m-1}{\ell+1}z^\ell
        }{(1-z)^m}.
    \end{align}
    Finally, using
    \begin{equation}
        (1-z)^{-m}
        =
        \sum_{r\geq0}
        \binom{m+r-1}{m-1}z^r
    \end{equation}
    and extracting the coefficient of $z^k$ proves
    \eqref{eq:odd-cardinality}.
\end{proof}

\section{Application to \texorpdfstring{{unitary designs}}{unitary designs}}\label{app:unitary-designs}

First, note that vectorization provides an embedding of $U(d)$ into
$\operatorname{End}(\mathbb{C}^{d})\simeq\mathbb{C}^{d^{2}}$:
\begin{equation}
    U\longmapsto
    \lvert U\rangle\!\rangle
    \coloneqq
    (U\otimes I)\lvert\Omega\rangle,
    \qquad
    \lvert\Omega\rangle
    \coloneqq
    \sum_{a=1}^{d}\lvert a,a\rangle .
\end{equation}
Under this identification, the Euclidean inner product on
$\mathbb{C}^{d^{2}}$ coincides with the Hilbert--Schmidt inner product,
\begin{equation}
    \langle\!\langle U\vert V\rangle\!\rangle
    =
    \operatorname{tr}\!\left[U^{\dagger}V\right].
\end{equation}
For simplicity, we henceforth write $\ket{U}\coloneqq\lvert
U\rangle\!\rangle$. This should cause no ambiguity, since no vectors in
the original space $\mathbb{C}^{d}$ will appear below.

Let
\begin{equation}
    \chi=\{(p_i,U_i)\}_{i=1}^{N},
    \qquad
    p_i>0,
    \qquad
    \sum_{i=1}^{N}p_i=1,
\end{equation}
be a (weighted) ensemble of unitaries. Its $k$-th unitary frame potential is
\begin{equation}
    \mathcal{E}_{k}(\chi)
    \coloneqq
    \sum_{i,j=1}^{N}
    p_i p_j
    \left|
        \braket{U_i}{U_{j}}
    \right|^{2k}.
\end{equation}
The unitary Welch bound states that
\begin{equation}
    \mathcal{E}_{k}(\chi)\geq C_{d,k},
\end{equation}
with equality if and only if $\chi$ is a {weighted unitary $k$-design (equivalently, a weighted projective $k$-design for the Haar measure on
$PU(d)$)} \cite{GrossAudenaertEisert2007,RoyScott2009}. The Haar constant
is
\begin{equation}
    C_{d,k}
    \coloneqq
    \int_{U(d)}
    \left|\operatorname{tr}U\right|^{2k}\mathrm{d}U
    =
    \sum_{\substack{\lambda\vdash k\\ \ell(\lambda)\leq d}}
    \left(f^{\lambda}\right)^{2},
\end{equation}
where $f^{\lambda}$ is the dimension of the irreducible representation
of $S_k$ associated with $\lambda$
\cite{DiaconisShahshahani1994,Rains1998}.

Define the Haar and empirical moment operators by
\begin{equation}
    \tau_k
    \coloneqq
    \int_{U(d)}\ketbra{U}^{\otimes k}\mathrm{d}U,
    \qquad
    M_k(\chi)
    \coloneqq
    \sum_{i=1}^{N}p_i\ketbra{U_i}^{\otimes k}.
\end{equation}
The same computation as above gives
\begin{equation}
    \left\|
        \tau_k-M_k(\chi)
    \right\|_{2}^{2}
    =
    \mathcal{E}_{k}(\chi)-C_{d,k}.
\end{equation}

Let $\mathsf{\Gamma}_m$ denote partial transposition on the last $m$ tensor factors, where $m\in\{0,\dotsc,k\}$.
Then
\begin{equation}
    M_k(\chi)^{{\mathsf{\Gamma}_m}}
    =
    \sum_{i=1}^{N}p_i
    \ketbra{U_i}^{\otimes(k-m)}
    \otimes
    \ketbra{\overline{U_i}}^{\otimes m},
\end{equation}
and hence
\begin{equation}
    \operatorname{rank}\!\left[M_k(\chi)^{{\mathsf{\Gamma}_m}}\right]\leq N.
\end{equation}
{Moreover, $M_k(\chi)^{\mathsf{\Gamma}_m}$ and $\tau_k^{\mathsf{\Gamma}_m}$ are positive semidefinite and have the same trace, $d^k$. Thus the rank--trace argument underlying Theorem~\ref{Th:StrongerWB} applies without change. Let}
\begin{equation}
    {r_m\coloneqq\operatorname{rank}\!\left[\tau_k^{\mathsf{\Gamma}_m}\right],
    \qquad
    \lambda_1^{(m)}\geq\cdots\geq\lambda_{r_m}^{(m)}>0}
\end{equation}
{be the ordered positive spectrum, and for $N<r_m$ set}
\begin{equation}
    {\Delta_{\ell,m}(N)
    \coloneqq
    \sum_{j=N+1}^{r_m}\left(\lambda_j^{(m)}\right)^\ell.}
\end{equation}
{Then}
\begin{align}
    {\mathcal E_k(\chi)-C_{d,k}}
    &{=\left\|M_k(\chi)^{\mathsf{\Gamma}_m}-\tau_k^{\mathsf{\Gamma}_m}\right\|_2^2}
    \\
    &{\geq \Delta_{2,m}(N)+\frac{\Delta_{1,m}(N)^2}{N}.}
\end{align}
{This holds for every $m$ with $N<r_m$, so the largest of these spectral-tail bounds may be used. Evaluating it explicitly requires the spectrum of $\tau_k^{\mathsf{\Gamma}_m}$.}

{How these bounds compare with those in the literature} \cite{BrandaoHarrowHorodecki2016,BrandaoEtAl2021,LancienMajenz2020,NakataEtAl2025,BrakerskiYuen2026} and whether an analogue of {Theorem~\ref{Th:VS-WB}} could also be proved in this context remain {nontrivial} open questions, which might be addressed in future research.

\end{document}